\documentclass[%
reprint,
superscriptaddress,
noeprint,
amsmath,amssymb,
longbibliography,
aps,
prx,
]{revtex4-1}
\usepackage{mathtools}
\usepackage{amsthm}
\newtheorem{theorem}{Theorem}%[section]
\newtheorem{definition}{Definition}%[theorem]
\newtheorem{lemma}{Lemma}
\newtheorem{example}{Example}%[section]
\newtheorem{prop}{Proposition}%[theorem]

\usepackage{dsfont}
\usepackage{physics,xcolor,comment}
\usepackage{wasysym}
\usepackage{graphicx,caption,subcaption}% Include figure files
\usepackage{dcolumn}% Align table columns on decimal point
\usepackage{bm}% bold math
\usepackage{hyperref}% add hypertext capabilities
\usepackage[capitalize]{cleveref}
\usepackage[mathlines]{lineno}% Enable numbering of text and display math

\def \mc #1{\mathcal{#1}}

\newcommand{\tN}{\tilde{\mathcal{N}}}
\newcommand{\Sp}[1]{\mathbf{v}(#1)}
\newcommand{\str}[1]{\text{sTr}\left[#1\right]}

\begin{document}

\preprint{APS/123-QED}

%\title{Theory of Error Mitigation}% Force line breaks with \\
%\thanks{A footnote to the article title}
\title{NISQ: Error Correction, Mitigation, and Noise Simulation}

\author{Ningping Cao}
\thanks{These authors contributed equally to this work.}
\affiliation{Institute for Quantum Computing and Department of Physics and Astronomy, University of Waterloo, Waterloo, ON N2L 3G1, Canada}
\author{Junan Lin}
\thanks{These authors contributed equally to this work.}
\affiliation{Institute for Quantum Computing and Department of Physics and Astronomy, University of Waterloo, Waterloo, ON N2L 3G1, Canada}
\author{David Kribs}
\affiliation{Department of Mathematics \& Statistics, University of Guelph, Guelph, ON N1G 2W1, Canada}
\author{Yiu-Tung Poon}
\affiliation{Department of Mathematics, Iowa State University, Ames, IA, USA 50011}
\author{Bei Zeng}
\email{zengb@ust.hk}
\affiliation{Department of Physics, The Hong Kong University of Science and Technology, Clear Water Bay, Kowloon, Hong Kong}
\affiliation{Institute for Quantum Computing and Department of Physics and Astronomy, University of Waterloo, Waterloo, ON N2L 3G1, Canada}
\author{Raymond Laflamme}
\affiliation{Institute for Quantum Computing and Department of Physics and Astronomy, University of Waterloo, Waterloo, ON N2L 3G1, Canada}

%\date{\today}

\date{\today}

\begin{abstract}
Error-correcting codes were invented to correct errors on noisy communication channels. 
Quantum error correction (QEC), however, has a wider range of uses, including information transmission, quantum simulation/computation, and fault-tolerance. 
These invite us to rethink about QEC, in particular, about the role that quantum physics plays in terms of encoding and decoding. 
The fact that many quantum algorithms, especially near-term hybrid quantum-classical algorithms, only use limited types of local measurements on quantum states, leads to various new techniques called Quantum Error Mitigation (QEM). 
This work examines the task of QEM from several perspectives.
Using some intuitions built upon classical and quantum communication scenarios, we clarify some fundamental distinctions between QEC and QEM.
We then discuss the implications of noise invertibility for QEM, and give an explicit construction called \textit{Drazin-inverse} for non-invertible noise, which is trace preserving while the commonly-used Moore-Penrose pseudoinverse may not be.
Finally, we study the consequences of having an imperfect knowledge about the noise, and derive conditions when noise can be reduced using QEM.

\end{abstract}

\maketitle

\tableofcontents

\section{Introduction}

The field of quantum information processing has entered an era featuring noisy, intermediate-scale quantum (NISQ) devices.
Despite some recent demonstrations of computational advantages compared to classical computers~\cite{Arute2019,Zhong2020}, NISQ devices still face significant challenges before eventually becoming practically useful.
In particular, noise in NISQ processors can spoil the computation process and possibly lead to incorrect final results.

Conventionally, the main tool for protecting the processor from noise has been quantum error correction (QEC).
QEC protocols are designed to allow a user to detect, and eventually correct, errors that happen during a quantum computation.
While many approaches for QEC have been developed, few have been tested on real quantum processors due to the significant demands on the hardware. 
First, QEC generally encodes quantum information into a much larger Hilbert space, which requires the hardware size to be large as well.
Second, quantum operations (gates) on a processor must be below a certain threshold value for QEC to successfully reduce the effective error, instead of introducing more errors.
Meeting both requirements is generally difficult on most state-of-the-art devices available today.

Recently, the field of quantum error mitigation (QEM) emerged with the goal of decreasing the effective noise level, while circumventing these two obstacles, on near term devices.
The idea is that if one has some knowledge about the noise processes in a particular hardware, then one should be able to utilize that knowledge to reduce (part of) the effect of that noise.
Importantly, it is more desirable to have protocols that introduce no or very little additional hardware overhead in order to improve the computation accuracy.
Numerous protocols have been developed during the past few years~\cite{Temme2017a,Endo2018,McArdle2019,Maciejewski2020,Koczor2021} that fall into this category.

The parallel development of both fields naturally leads to the question: under what circumstances should one apply QEC over QEM, and vice versa?
To answer this question, it is useful to clearly illustrate the differences between QEC and QEM.
Understanding these differences can be helpful during the experimental design stage, where an experimentalist decides which protocol to use to achieve a particular task.

In this work, we examine the relation between QEC and QEM from a high-level perspective, and discuss several implications when using QEM in practice.
After reviewing some common QEM protocols in \cref{sec_review_QEM}, in \cref{sec_QEC_vs_QEM} we illustrate the fundamental differences between QEC and QEM from the perspective of classical and quantum communication.
In \cref{sec_Drazin} we study implications of noise invertibilities in QEM, and illustrate how non-invertible noise, which is largely omitted in the literature, may arise in an experiment.
We propose a construction called Drazin-inverse for non-invertible noises, and prove that compared to a conventional choice of pseudoinverse (the Moore-Penrose pseudoinverse), the Drazin-inverse has the advantage of being trace preserving, which can have better stability when running computer simulations.
In \cref{sec_conseq} we study the consequences due to imperfect knowledge about actual noise, and give a sufficient condition for when an optimal QEM can improve the expectation value of any observable.

\section{Review of QEM protocols}\label{sec_review_QEM}
The goal of QEM protocols is usually to recover the ideal expectation value of some observable $A$ from the output state $\rho_{\text{out}}$, namely $\langle A \rangle = \Tr[A \rho_{\text{out}}]$.
These can be further classified into two categories~\cite{Yoshioka2021,Zhang2021}.
The first one, which we call error-based QEM, requires prior knowledge on the form of noise occurring in a quantum processor.
The second one does not require the particular form of noise being known, and will be called error-agnostic QEM.
This section aims to review some of the commonly-used protocols, with a focus on error-based QEM.
Our treatment will be based on the review article by Zhang et al.~\cite{Zhang2021} with some modifications.

\subsection{Extrapolation Methods}
Extrapolation-based QEM is among the earliest proposed protocols that fall under the QEM category~\cite{Temme2017a, Li2017}.
The intuition is that while it is in general difficult to reduce the physical error rate on a processor, \emph{increasing} the error rate in a controlled manner might be possible on some systems.
Once sufficiently many observations of the expectation value under different noise strengths have been obtained, it could be possible to infer the noiseless value from the noisy ones.
Consider a stochastic noisy physical process of the form
\begin{equation}
\mc{E} = (1-\epsilon) \mc{I} + \epsilon \mc{N}
\end{equation}
where $\mc{I}$ is the identity map, $\mc{N}$ is the noise map, and $\epsilon$ is a small error rate.
For a target quantum circuit with depth $n$, the ideal and noisy output states can be written as
\begin{equation}\label{eq_n_layer}
	\begin{gathered}
	\rho_\text{out}^\text{ideal} = \mathcal{U}_n\circ\cdots\circ\mathcal{U}_1(\rho_\text{in}),\\
	\rho_\text{out}^\text{exp} = \mc{N} \circ \mathcal{U}_n\circ\cdots\circ  \mc{N} \circ\mathcal{U}_1(\rho_\text{in}),
	\end{gathered}
\end{equation}
where we have made a simplifying assumption that the noise is Markovian and gate-independent.
The expectation value can be expanded in a series as
\begin{equation}
\langle A \rangle(\epsilon) = \langle A \rangle(0) + \sum_{m=0}^{k} A_{m} \epsilon^{m} +\mc{O}(\epsilon^{m+1}),
\end{equation}
where $A_{m}$ are expansion coefficients, and $\langle A \rangle(\epsilon)$ denotes the expectation value corresponding to an error rate of $\epsilon$.
One can then perform a series of experiments with different noise strength $\{\lambda_{i} \epsilon\}$, where $i=0,...,k$ and $\lambda_{i}>0$, and collect the expectation values $\langle A \rangle(\lambda_{i} \epsilon)$.
Applying the Richardson extrapolation method, one obtains the parameters $r_{i}$ to infer the noiseless value up to precision $\mc{O}(\epsilon^{k+1})$:
\begin{equation}
\langle A \rangle_{\text{EM}} = \sum_{i=0}^{k} r_{i} \langle A \rangle(\lambda_{i} \epsilon) = \langle A \rangle(0) + \mc{O}(\epsilon^{k+1}).
\end{equation}

Besides the original approach using Richardson extrapolation, other extrapolation methods exist.
For example, the exponential or linear extrapolation method used by Endo et al. \cite{Endo2018} uses two data points only.
We note that extrapolation methods are accurate up to order $\mc{O}(\epsilon^{k+1})$, compared with quasiprobability method below which can be fully accurate in principle.

\subsection{Quasiprobability Methods}
QEM using quasiprobability sampling was also proposed by Temme et al. in the same paper where Richardson extrapolation QEM was introduced~\cite{Temme2017a}, and was later generalized by Endo et al. in~\cite{Endo2018}.
Consider a target unitary gate $\mc{U}$ and its noisy version $\mc{U}_{\text{noisy}} = \mc{N} \mc{U}$, i.e., the ideal gate followed by a noise process.
One can always find a CPTP map $\mc{N}$ that satisfies the above equation, since the inverse of $\mc{U}$ is always CPTP, and the composition of two CPTP maps is CPTP.
The assumption is that there exists a complete set of (noisy) operations available to the experimentalist, denoted by
\begin{equation}
	\{\mc{G}_{1}, \dots, \mc{G}_{K} \}
\end{equation}
which has been pre-characterized.
This set is complete in the sense that they form a basis for the inverse noise channel $\mc{U}_{\text{noisy}}^{-1}$,
\begin{equation}
	\mc{N}^{-1} = \sum_{i} a_{i} \mc{G}_{i} = \tau \sum_{i} \frac{\abs{a_{i}}}{\tau} \text{sgn}(a_{i}) \mc{G}_{i}
\end{equation}
where $\tau = \sum_{i} \abs{a_{i}}$.
The coefficients $p_{i} \coloneqq \frac{a_{i}}{\tau}$ now form a probability distribution, so one can append additional gates $\mc{G}_{i}$ with probability $p_{i}$ after the noisy circuit to reverse the effects of $\mc{N}$.
For an observable $A$, we can express its ideal value by 
\begin{equation}
	\langle A \rangle_{\text{EM}} = \Tr[\mc{N}^{-1} \mc{N} \mc{U}(\rho) A] = \tau \sum_{i} p_{i} \text{sgn}(a_{i}) \langle A \rangle_{i}
\end{equation}
where $\langle A \rangle_{i} = \Tr[\mc{G}_{i} \mc{N} \mc{U} (\rho) A]$.
We see that by adding in $\mc{G}_{i}$ at the end of the circuit with probability $p_{i}$, and keeping track of each $\text{sgn}(a_{i})$, one obtains the ideal expectation value.
Note that $\sum_{i} a_{i} = 1$ due to the trace-preserving constraint.
Since some of the $a_{i}$'s are negative, $\tau \geq 1$, so the $\langle A \rangle_{\text{EM}}$ has a variance that is approximately $\tau^{2}$ times larger than that measuring with the ideal circuit.

It is interesting to point out that the original proposal for quasiprobability QEM aimed not to append $\mc{N}^{-1}$ at the end, but to directly simulate the \emph{ideal} unitary $\mc{U}$.
It may appear that this scenario falls under the error-agnostic category.
However, since it is still required that all noisy gates be characterized precisely, knowledge of noise in the real system is necessary.
Therefore the method is still implicitly error-based.
The effects of noise characterization comes in nontrivially through the gate decomposition procedure.
In practice, it is sometimes more favorable to use the $\mc{N}^{-1}$ approach introduced here.
This may happen when, for example, a decomposition for the near-identity map $\mc{N}$ has a much smaller cost $\tau$ than that for the target unitary $\mc{U}$.

\subsection{Readout Error Mitigation}
Readout error mitigation~\cite{Maciejewski2020,Chen2019b} aims to reduce the effect of noisy measurement operations during the readout step of quantum computation.
A general $K$-outcome measurement process can be described by a set of positive operator-valued measure (POVM) elements,
\begin{equation}
	\{M_{1}, \dots, M_{K}\},
\end{equation}
such that the probability of obtaining outcome $j$ given an input state $\rho$ is $p_{j} = \Tr[M_{j} \rho]$.
We may arrange the outcome probabilities with an ideal measurement apparatus as a vector, $P_{\text{ideal}} = (p_{1}, \dots, p_{K})^{T}$, and similarly for the actual outcome probability vector $P_{\text{noisy}}$.
One can see that there exists a transformation $T$ between the two such that
\begin{equation}
	P_{\text{noisy}} = T \cdot P_{\text{ideal}}.
\end{equation}
The matrix $T$ can be learned through experimentally measuring different input states, and its number of parameters may be less than the size of $T$ if one makes further assumptions about the form of noise~\cite{Bravyi2021}.
Once $T$ is learned, the ideal probability vector can be obtained by $P_{\text{ideal}} = T^{-1} \cdot P_{\text{noisy}}$.

\subsection{Error-Agnostic QEM}
We have so far focused on error-based QEM protocols.
Below we briefly go over some other protocols that do not require explicit knowledge about the form of error, and are therefore error-agnostic.
These approaches generally utilize certain structures of the problem that is known to the user prior to conducting any experiment.
Prominent examples include virtual distillation (VD), symmetry verification (SV), and $N$-representability.
In VD~\cite{Koczor2021,Huggins2021,Huo2022}, one obtains effective error-mitigated expectation values of an observable $A$ as $\langle A \rangle_{\text{EM}} = \Tr[A \rho_{\text{VD}}^{(m)}]$ where $\rho_{\text{VD}}^{(m)} = \rho^{m}/\Tr[\rho^{m}]$, through entangling operations between $m$ copies of the output state $\rho$.
The purity of this effective state increases with $m$, thereby eliminating the effects of stochastic errors.

In SV~\cite{Bonet-Monroig2018,McArdle2019}, one examines certain symmetry properties of the output state from a quantum program, and performs post-selection to remove outputs that disobey these symmetries.
For example, in variational quantum eigensolver (VQE) algorithms, one prepares an ansatz input state which is then updated variationally, in order to approach the true ground state of a target Hamiltonian.
There are often certain symmetries that are obeyed by the true ground state, such as the total number of electrons, as well as the number of spin-up spin-down electrons.
One can thus measure these symmetries in the output state, and discard results that violate any symmetry tested.

The idea of $N$-representability QEM~\cite{Rubin2018} is similar to that of SV.
Its name originated from the $N$-representability conditions in quantum chemistry, which stand for a set of necessary conditions that must be satisfied by the reduced (marginal) density matrices.
Due to the presence of noise, states measured in experiments might violate some of these conditions, so one may reduce the effects of these noise by projecting the experimental state onto the closest state satisfying the $N$-representability conditions.

\section{Error Correction vs. Mitigation: a Communication Viewpoint}\label{sec_QEC_vs_QEM}
One can see from the previous review that many protocols with different nature are currently summarized under the name ``quantum error mitigation''.
In this section we will provide a high-level comparison between QEM and QEC, with a focus on error-based QEM.
As previously shown, these are procedures that utilize knowledge about the noise occurring during computation and \emph{actively} try to eliminate their effects, making them more similar to \emph{active} QEC.
On the other hand, error-agnostic QEM are, in the eyes of the authors, closer to passive error correction.
Nonetheless, insightful comparison between error-agnostic QEM and QEC can be an interesting future research direction and may facilitate integrations between QEC and QEM in general.

The development of QEC has historically stemmed from previous studies of classical error correction (CEC).
Therefore, before discussing QEC or QEM, it is instructive to first look at their classical counterparts.
We will start our discussion by considering a communication task, which is essentially a trivial computation with the target unitary being the identity, but with two spatially separated parties being the input and output ends.
This illustrates some simple yet important distinction between EC and EM.
We then move on to computation and discuss the usefulness of QEM in quantum computation, using our intuition from the communication case.

\subsection{Classical Communication}\label{subsec_c_comm}

\begin{figure}[ht]
	\centering
	\includegraphics[width=\linewidth]{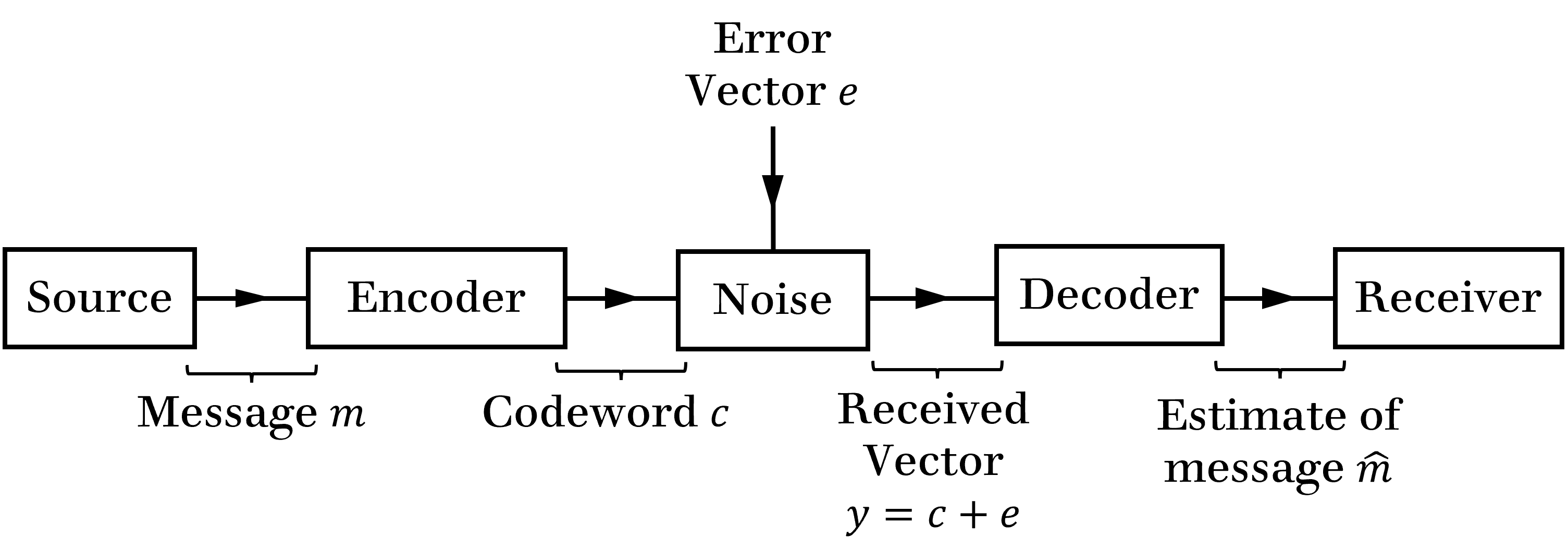}
	\caption{An illustrative diagram for noisy classical communication.}
	\label{fig_ccomm}
\end{figure}

In a classical communication task, a sender Alice would like to transmit a $k$-bit string to Bob.
An example of the string would look like
\begin{equation*}
s = 11010101000100110101010011...
\end{equation*}
Alice and Bob share a classical communication channel $C$ which is subject to noise.
She does so by sending each letter of the text through $C$, involving a total of $k$ uses of $C$.
For simplicity, assume that the noise is described by a binary symmetric noise channel with strength $p$, denoted as $\text{BSC}_{p}$.
This channel preserves the sent bit with probability $1-p$, and flips it (symmetrically from $0$ to $1$ and from $1$ to $0$) with probability $p$.
The matrix representation of the noise is given by
\begin{equation}\label{eqn_BSC_noise}
\text{BSC}_{p} \rightarrow \begin{pmatrix}
1-p & p \\ p & 1-p
\end{pmatrix} \coloneqq N_{\text{BSC},p}
\end{equation}
where both the input and output basis are ordered as $\{0,1\}$.
If Alice sends a bit $0$, it will have a distribution
\begin{equation}
N_{\text{BSC},p} \begin{pmatrix}
1 \\ 0
\end{pmatrix} = \begin{pmatrix}
1-p \\ p
\end{pmatrix},
\end{equation}
at the output end.

%Classical error correcting codes have been developed to fight this noise.
An illustrative diagram for classical error correcting codes against this noise is given in \cref{fig_ccomm}.
The simplest example is the $3$-bit repetition code, defined by the encoding  
\begin{equation}\label{eqn_3_bit_code}
0 \rightarrow 000,\ 1 \rightarrow 111,
\end{equation}
i.e., each bit is repeatedly encoded $3$ times.
The number of uses of the channel has now increased $3$ times to $3k$.
The decoding is done by performing a majority vote on the received bits, so that at the receiving end, Bob again obtains a bit string of length $k$.
Assuming that $\text{BSC}_{p}$ acts independently on each bit, the probability of error is reduced from $p$ to $3p^{2}(1-p) = \mc{O}(p^{2})$ by this code.

Next we ask the question: Can Bob improve the quality of communication, given some knowledge about the noise?
From the form of noise in \cref{eqn_BSC_noise}, the best possible knowledge Bob could have is the precise value of $p$.
Suppose in addition that $p<1/2$.
If Bob receives a bit $1$, then he only knows that Alice more likely sent a $1$ than a $0$, so the best deterministic procedure is simply to keep the bit.
Applying this argument to all received bits, we see that the best Bob can do is to simply keep all received bits intact.
Similarly, if $p>1/2$, then Bob's best action is to flip all the received bits.
Clearly, this does not increase Bob's information on Alice's message, as measured by the classical mutual information.
In fact, what we have shown is a special case of the classical data-processing inequality, which states that no post-processing can increase the mutual information between Alice and Bob.

What Bob \emph{can} do is to recover the \emph{distribution} of Alice's input for sufficiently large $k$.
Specifically, he can apply the inverse map of $N_{\text{BSC},p}$ on his received distribution, resulting in
\begin{equation}
N_{\text{BSC},p}^{-1} N_{\text{BSC},p} v_{A} = v_{A},
\end{equation}
where $v_{A} = (p_{0,a}, p_{1,a})^{T}$ is Alice's input distribution.
However, Bob cannot further use this restored distribution to recover Alice's \emph{message}.
The best he can do is to use the restored distribution to randomly generate a new $k$-bit string, during which Alice's message is completely destroyed.

\subsection{Quantum Communication}\label{subsec_q_comm}

\begin{figure}
	\centering
	\begin{subfigure}[b]{0.42\linewidth}
		\centering
		\includegraphics[width=\textwidth]{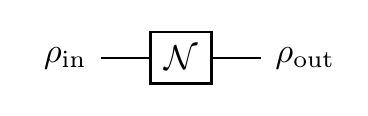}
		\caption{}
		\label{fig_qcomm_1}
	\end{subfigure}
	\hfill
	\begin{subfigure}[b]{0.65\linewidth}
		\centering
		\includegraphics[width=\textwidth]{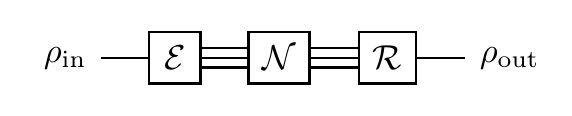}
		\caption{}
		\label{fig_qcomm_2}
	\end{subfigure}
	\caption{(a) Quantum communication model where Alice sends a state $\rho_{\text{in}}$ to Bob through a noisy channel $\mc{N}$. (b) An attempt to reduce the effect of $\mc{N}$ using QEC, through an encoding operation $\mc{E}$ and a recovery operation $\mc{R}$.}
	\label{fig_qcomm}
\end{figure}

In quantum communication, Alice and Bob communicates by sharing a quantum channel $Q$ capable of transmitting quantum particles.
Alternatively, one may also model the situation as having a central source sending out (potentially entangled) particles to Alice and Bob.
If this source is located in Alice's lab, then the channel from source to Alice is ideal and the one to Bob is described by a noisy channel $\mc{N}$, as shown in \cref{fig_qcomm_1}.
Here, they are interested in either encoding classical information in the particles and use $Q$ to transmit a classical message, or directly sharing quantum particles which may have been prepared in some special states.

\subsubsection{Quantum Communication of Classical Information}\label{QC_classical_info}
In the first scenario, the classical information may be extracted at Bob's end by measuring the received state.
For each use of the channel, we will name the state that Alice sent $\rho_{\text{in}}$ and the one coming out of Bob's end $\rho_{\text{out}} = \mc{N}(\rho_{\text{in}})$.
Without loss of generality, assume that the information is encoded in the expectation value of some observable $O$.
Due to the nature of quantum measurements, every possible outcome occurs randomly with probabilities determined by the Born rule.
Therefore, Alice must send multiple copies of the state in order for Bob to extract the expectation value, $\Tr[O \rho_{\text{out}}]$.
This is to be contrasted with classical communication, where sending 1 bit of information only involves using $C$ once, in the limit of an ideal channel.
One thus sees that, sending classical information using a quantum channel inevitably results in Bob measuring a distribution of possible outcomes.

If Bob knows the exact form of $\mc{N}$, then he may reconstruct the matrix  $\rho_{\text{out}}$ using sufficiently many copies of received states, assuming that the size of $\rho_{\text{out}}$ is not too large.
He may then apply $\mc{N}^{-1}$ to perfectly restore $\rho_{\text{in}}$, and consequently use this classical copy to infer the expectation value of any observable.
This is in direct analogy to restoring the distribution of input characters using CEM described in the previous section.

Alternatively, Alice and Bob may also use QEC to fight against this noise.
A QECC is defined by an encoding scheme, which is a completely positive and trace preserving (CPTP) map $\mc{E}: \mathds{C}_{2^{k}} \rightarrow \mathds{C}_{2^{n}}$ (where $n > k$), and then decode at Bob's end using another CPTP map $\mc{R}: \mathds{C}_{2^{n}} \rightarrow \mathds{C}_{2^{k}}$.
Here $\mathds{C}_{2^{n}}$ denotes the complex Euclidean space with dimension $2^{n}$.
This involves a total of $nk$ uses of the quantum channel.
Assuming that $\mc{N}$ is correctable by the QECC~\cite{Knill1996}, the output state $\rho_{\text{out}}$ will be equal to the input $\rho_{\text{in}}$, so directly measuring this error-corrected state will give Bob the correct expectation value for any observable.

One may now recognize the similarity between QEM and CEM for reconstructing the input \emph{distribution}.
Indeed, a density matrix \emph{is} probabilistic description of outcomes of any possible measurement on a quantum system.
In this scenario, it is sufficient for Bob to reconstruct $\rho_{\text{in}}$ as a mathematical object in order to eliminate the effect of $\mc{N}$, since this determines the outcome of \emph{any measurement} Bob can possibly make on the system.

\subsubsection{Quantum Communication of Quantum Information}

\begin{figure}
	\centering
	\begin{subfigure}[b]{0.4\linewidth}
		\centering
		\includegraphics[width=\textwidth]{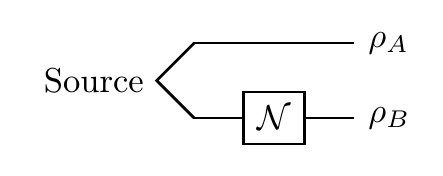}
		\caption{}
		\label{fig_1EPP_1}
	\end{subfigure}
	\hfill
	\begin{subfigure}[b]{0.55\linewidth}
		\centering
		\includegraphics[width=\textwidth]{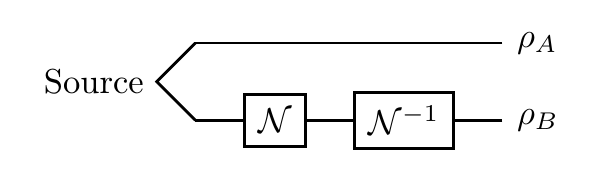}
		\caption{}
		\label{fig_1EPP_2}
	\end{subfigure}
	\hfill
	\begin{subfigure}[b]{0.9\linewidth}
		\centering
		\includegraphics[width=\textwidth]{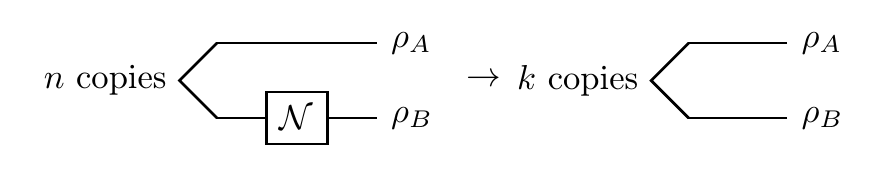}
		\caption{}
		\label{fig_1EPP_3}
	\end{subfigure}
	\caption{(a) Figure where Alice and Bob receives an EPR pair $\Phi^{+}$ from a source, and a noise $\mc{N}$ occur for Bob's channel. (b) A QEM approach. (c) A QEC approach.}
	\label{fig_1EPP}
\end{figure}

In the second scenario, the goal is to send particles encoding quantum information from Alice to Bob.
The physical states of the particles represent the quantum information encoded, which must be kept ``alive'' for such a task.
This scenario is relevant when, for example, $\rho_{\text{in}}$ is the output from a computation done by Alice, which needs to be transferred to Bob for further processing.
Clearly, the QEC approach shown in \cref{fig_qcomm_2} is capable of eliminating the effect of $\mc{N}$ if it falls within the set of correctable errors, because what Bob gets at the end is an actual quantum object.
In contrast, one cannot use the ``error mitigation'' approach described in the previous section, where Bob applies $\mc{N}^{-1}$ onto a \emph{classical image} of $\rho_{\text{out}}$, since this only corrects Bob's description about $\rho_{\text{out}}$ rather than the physical state itself.

Another case is when Alice and Bob tries to establish shared entanglement to achieve certain tasks, such as winning a nonlocal game.
This is relevant because entanglement is among the most commonly accepted benchmarks for quantum information, and was speculated to have a similar role as classical information~\cite{Bennett1996}.
As shown in \cref{fig_1EPP}, Alice prepares $k$ copies of maximally entangled Bell pairs
\begin{equation}
\Phi^{+} = (\ket{00}_{AB} + \ket{11}_{AB})/\sqrt{2}
\end{equation}
and sends half of each pair to Bob.
Again, the noisy channel $\mc{N}$ to Bob potentially reduces the entanglement shared between Alice and Bob.
Historically, this task is achieved by a family of procedures called entanglement purification protocols (EPPs)~\cite{Bennett1996}, illustrated in \cref{fig_1EPP_3}.
In EPP, Alice and Bob needs to start from $n>k$ copies of the noisy Bell state, and obtain $k$ pairs at the end which are closer to the pure state $\Phi^{+}$. 
Profoundly, a sub-class of EPP protocols called one-way EPP (or 1-EPP), where only communication from Alice to Bob is allowed, have a one-to-one correspondence with QECCs~\cite{Bennett1996}.
Therefore, QEC is again useful in this scenario.

On the contrary, one sees that the QEM approach described in the previous section cannot be used to purify entanglement.
In particular, all entanglement between Alice and Bob would be destroyed due to the measurement.
In fact, analogous to the classical case where Bob recovers Alice's input distribution and generate a random $k$-bit string, here Bob knows in advance that he will ideally get a maximally mixed state $I/2$; so the above protocol is simply equivalent to Bob generating $(I/2)^{k}$ locally, and discarding all qubits received from the source!

The above example illustrates a fundamental distinction between QEM, which is only capable of restoring the classical image of a quantum system, and QEC, which is capable of restoring the quantum object \emph{itself}, along with all possible non-classical resources possessed by that object.
It is instructive to recall again the case where one needs to preserve classical information (see \cref{subsec_c_comm}), where we have also argued that CEC is helpful for such a task, while CEM is not.
Furthermore, recall in \cref{QC_classical_info} we argued that recovering density matrices in QEM is analogous to recovering classical distributions in CEM.
These complete our comparisons between EM and EC, which are summarized in \cref{table_ec_em_comp}.

\begin{table*}[ht]
	\centering
	\begin{tabular}{c | c | c}
		\hline
		& EC & EM \\
		\hline
		Classical & Classical information being transmitted & Classical distribution of possible outputs \\
		\hline
		Quantum & Physical quantum objects being transmitted & Density matrices describing the physical quantum objects\\
		\hline
	\end{tabular}
	\caption{Summary of what EC and EM could recover at the output end, under classical and quantum communication settings.}
	\label{table_ec_em_comp}
\end{table*}

\section{Noise Invertibility and the Drazin-Inverse}\label{sec_Drazin}

Having illustrated the fundamental distinction between error-based QEM and QEC under the communication setting, it is then natural to ask under what circumstances are these two methods useful.
For QEC, the set of errors (and linear combinations within) are correctable if they obey the Knill-Laflamme conditions~\cite{Knill1997,Knill1996} for the particular code being used.
While there exists no parallel, quantitative results for QEM, the form of some error-based QEM protocols suggests that the inverse noise channel plays an important role, since it is what these protocols aim to apply to the output (directly in the case of readout QEM, and indirectly in the case of quasiprobability QEM, for instance).
In this section, we discuss the implications of noise invertibility on error-based QEM, and propose an alternative quasi-inverse construction in the case of non-invertible noise.

A quantum noise process is, on the physical level, described by a completely-positive (CP) and trace-preserving (TP), or CPTP map, $\mc{N}$.
Below we first define a matrix representation for quantum states and maps.
In this work, we denote the space of linear operators mapping Hilbert space $H_{A}$ to $H_{B}$ as $L(H_{A},H_{B})$, or $L(H_{A})$ in short if $H_{A} = H_{B}$.
Let $T(H_{A},H_{B})$ be the space of linear maps from $L(H_{A})$ to $L(H_{B})$.
Let $e_{i}$ be the standard basis of $H_{i}$ with a 1 at position $i$ and 0 elsewhere. 
Let $E_{a,b}$ be the standard basis of $L(H_{A},H_{B})$ with a 1 at position $(a,b)$ and 0 elsewhere. 
\begin{definition}
	(Vectorization of linear operators.) The vec mapping $\Sp{\cdot}:\ L(H_{A},H_{B}) \rightarrow H_{B} \otimes H_{A}$ is the unique mapping that satisfies $\Sp{E_{a,b}} = e_{b} \otimes e_{a}$. 
\end{definition}
Next we define two representations for quantum maps.
\begin{definition}
	(Choi representation.) The Choi representation of a map $\mc{M} \in T(H_{A},H_{B})$ is defined by $C(\mc{M}) = \sum_{a,b}  E_{a,b} \otimes \mc{M}(E_{a,b})$.
\end{definition}

\begin{definition}
	(Natural representation.) The natural (or equivalently, superoperator) representation of a map $\mc{M} \in T(H_{A},H_{B})$ is defined by the unique linear operator $\Sp{\mc{M}} \in L(H_{A} \otimes H_{A}, H_{B} \otimes H_{B})$ that satisfies $\Sp{\mc{M}} \Sp{A} = \Sp{\mc{M}(A)}$ for all $A \in L(H_{A})$.
\end{definition}	

In the natural representation, the channel $\mc{N}$ acting on a quantum state $\rho$ can be written as the superoperator $\Sp{\mc{N}}$ multiplying the vector representation $\Sp{\rho}$ of the quantum state $\rho$~\cite{watrous2018theory}. The vector representation $\Sp{\rho}$ of $\rho$ inherits its ordering from the superoperator, hence we abuse the notation $\Sp{\cdot}$ for vector representations of quantum states and observables (which are often written as double kets $|\rho\rangle\!\rangle$ in other literature).

The following theorem directly comes from representation theory of linear maps.
\begin{theorem}\label{thm_invert}
	The quantum channel $\mc{N}$ is invertible iff $\Sp{\mc{N}}$ is an invertible matrix.
\end{theorem}

Below we give an example where the inverse $\mc{N}^{-1}$ of a CPTP map $\mc{N}$ is non-CP.
\begin{example}
	Let the Choi representation of a quantum channel $\mc{N}$ be
	\[
	C(\mc{N}) = 
	\left(\begin{array}{cc|cc}
	\frac34 & 0 & -\frac{i}8 & \frac12 + \frac{i}8\\ 
	0 & \frac14 & -\frac{i}8 & \frac{i}8 \\ \hline
	\frac{i}{8} & \frac{i}8 & \frac14 & 0\\
	\frac12 - \frac{i}8 & -\frac{i}8 & 0 & \frac34
	\end{array}\right).
	\]
	The superoperator is 
	\[
	\Sp{\mc{N}} = 
	\begin{pmatrix}
	%\left(\begin{array}{cc|cc}
	\frac34 & \frac{i}{8} & -\frac{i}{8} & \frac{1}{4}\\ 
	0 & \frac12 - \frac{i}8 & -\frac{i}{8} & 0\\
	0 & \frac{i}{8} & \frac12 + \frac{i}8 & 0\\
	\frac14 & -\frac{i}8 & \frac{i}{8} & \frac34
	%\end{array}\right).
	\end{pmatrix}.
	\]
	
	Therefore, the inverse of $\Sp{\mc{N}}$ is 
	
	\[
	\Sp{\mc{N}^{-1}} = 
	\left(\begin{array}{cccc}
	\frac32 & \frac14-\frac{i}2 & \frac14+\frac{i}2 & -\frac12 \\
	0 & 2+\frac{i}2 & \frac{i}2 & 0 \\
	0 & -\frac{i}2 & 2-\frac{i}2 & 0 \\
	-\frac{1}2 & -\frac{1}4+\frac{i}2 & -\frac{1}4-\frac{i}2 & \frac{3}2
	\end{array}\right).
	\]

Its Choi representation is
\[
C(\mc{N}^{-1}) = 
\left(\begin{array}{cc|cc}
\frac{3}2 & 0 & \frac14 + \frac{i}2 & 2-\frac{i}2\\ 
0 & -\frac{1}2 & \frac{i}2 & -\frac{1}4-\frac{i}2\\ \hline
\frac{1}4 - \frac{i}2& -\frac{i}2 & -\frac{1}2 & 0\\
2+\frac{i}2 & -\frac14+\frac{i}2 & 0 & \frac{3}2
\end{array}\right).
\]
The Choi representation $C(\mc{N}^{-1})$ has negative eigenvalues. Therefore, $\mc{N}^{-1}$ is a Hermitian preserving (HP) and trace preserving (TP) map, but not CP.
\end{example}

There are three distinct possibilities regarding noise invertibility.
The first is that $\mc{N}$ is invertible, and $\mc{N}^{-1}$ is CPTP.
In this case, the inverse $\mc{N}^{-1}$ is unique, and is Hermitian preserving (HP) and trace preserving (TP)~\cite{jiang2020physical}. 
Note that the set of CPTP maps is a subset of HPTP maps.
For $\mc{N}^{-1}$ to be CPTP, the channel $\mc{N}$ has to be a unitary channel, or can be seen as an unitary channel acting on the input state along with an ancilla prepared to a fixed state~\cite{nayak2006invertible}.
If the dimensions of input and output space are the same, the channel has a CPTP inverse \textit{iff} the channel is an unitary channel~\cite{preskill1998lecture,nayak2006invertible}.
This is a relatively easy scenario, because the quantum information can be coherently restored by physically applying $\mc{N}^{-1}$ to the output state.

The second possibility is that $\mc{N}$ is invertible, but $\mc{N}^{-1}$ is not CPTP.
A condition for when this will happen is later given in Proposition~\ref{thm:non-p}.
Many experimentally relevant noise models, such as the phase damping channel and the depolarizing channel, fall under this category.
Since $\mc{N}^{-1}$ is not a physically realizable operation, it cannot be experimentally implemented on the target system, so our above method to restore quantum information without redundancy fails.
Using QEM procedures, one can still recover the classical information in principle, by first extracting the classical output density matrix through measurements, and numerically apply the inverse map $\mc{N}^{-1}$.
But the process of measurement will inevitably disturb the system being measured, and destroy any entanglement it possibly has with other systems.

The third possibility is that $\mc{N}$ is non-invertible.
First, what does it mean for a quantum channel to be non-invertible?
Any CPTP map can be dilated to a unitary channel in a larger Hilbert space, and unitary channels are all invertible. 
What happens while tracing out the environment?
For a physical action to be written as a CPTP map, one needs to assume that the system and environment are separable in the beginning, which means that the input state has to have the form $\rho\otimes \sigma$.
If the unitary $U$ is also separable, tracing out the environment will also result in a unitary evolution on the system.
The correlation between system and environment causes the evolution of the system to become a general CPTP map. 
The non-invertibility of a channel signals strong non-locality, i.e., there is information in the whole system (system and environment) that is entirely invisible in the local system.

An example is the CNOT gate in~\cref{fig:cnots}. 
Let the initial state $\rho_{AB}$ be $\rho_A\otimes \ketbra{0}{0}$. Considering the controlling qubit $B$ as the environment, the channel ${\cal N}_A$ on the system qubit $A$ is non-invertible. The density matrix of qubit $A$ after the first CNOT gate is
\[
\rho_{A'} = {\cal N}_A(\rho_A) = \tr_{B}[U_{CNOT}(\rho_A\otimes \ket{0}\bra{0})U_{CNOT}^\dag]. %= \sum_{i=1}^2 E_i\rho_A E_i^\dag
\]
Some information (more precisely, certain off-diagonal entries) in the 2-qubit density matrix $\rho_{A'B'}$ does not reflect in the density matrix $\rho_{A'}$ of the system qubit (in fact, the non-local information cannot be seen locally in $\rho_{B'}$ as well).
However, after the second CNOT gate, the whole system $AB$ backtracks to the original state. 
That is to say
\begin{align*}
	\rho_A &= \rho_{A''} \\
	&= \tr_{B}[U_{CNOT}^\dag U_{CNOT}(\rho_A\otimes \ket{0}\bra{0})U_{CNOT} U_{CNOT}^\dag] \\
	&\neq {\cal N}_A^{-1}{\cal N}_A (\rho_A).
\end{align*}
This is because the second CNOT gate utilizes non-local information that is not available locally for qubit $A$.

Although the non-local information leak out of the system through the null space cannot be recovered without collecting information from the environment, the local rotation and contraction caused by the noise channel can be restored. 

\begin{figure}[ht]
	\includegraphics[width = .5\linewidth]{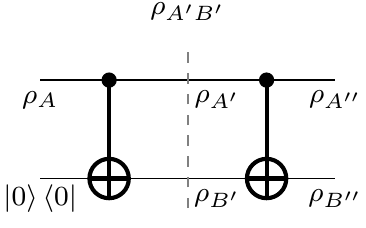}
	\caption{%The qubit $A$ is the system and qubit $B$ is the environment. 
	The equivalent channel $\cal{N}_A$ on system $A$ is non-invertible. However, the second CNOT gate brings the whole device back to the original state, i.e. $\rho_{A''B''} = \rho_{A''}\otimes\rho_{B''} = \rho_A\otimes \ket{0}\bra{0}$. The locally unseen information in $\rho_{A'B'}$ flow back to system $A$ after the second CNOT.}\label{fig:cnots}
\end{figure}

Second, how can we deal with these non-invertible noise channels?
It is known that the superoperator $\Sp{\mc{N}^{-1}}$ of the inverse channel $\mc{N}^{-1}$ equals to the inverse $\Sp{\mc{N}}^{-1}$ of the superoperator $\Sp{\mc{N}}$.
However, if the channel $\mc{N}$ is not invertible, the generalized inverse of $\Sp{\mc{N}}$ is not unique.
A commonly used generalized inverse is the Moore-Penrose inverse~\cite{moore1920reciprocal,penrose1955generalized}, but in~\cref{exp:pseudo} we show that the Moore-Penrose inverse of a CPTP map is not necessarily TP. 

In the following, we provide a construction of inverse-like channel $\mc{N}^+$.

Let the dimension of input and output space be $d$.
Take the Jordan decomposition of the superoperator of $\mc{N}$,
\begin{equation}\label{eq:decomp}
\mathbf{v}(\mc{N}) = Q\cdot J\cdot Q^{-1}
\end{equation}
where $J = \oplus_i J_{\lambda_i} $ is the Jordan normal form, $J_{\lambda_i}$ is a Jordan block corresponding to the eigenvalue $\lambda_i$, and $Q$ is a invertible matrix contains the generalized eigenvectors of $\Sp{\mc{N}}$.
If $\mathbf{v}(\mc{N})$ is diagonalizable, the Jordan normal form $J = \text{diag}[\lambda_1,\cdots,\lambda_{d^2}]$ is the diagonal matrix contains eigenvalues $\lambda_i$ of $\mathbf{v}(\mc{N})$.

We take the inverse-like channel $\mc{N}^{+}$ to be
\begin{equation}\label{eq:inverse}
\mathbf{v}(\mc{N}^{+}) = Q \cdot J'\cdot Q^{-1}.
\end{equation}
If $\mathbf{v}(\mc{N})$ is diagonalizable, $J'$ is the diagonal matrix that leaves the $0$'s in $J$ untouched and take the reciprocal of the rest elements in $J$.
If $\mathbf{v}(\mc{N})$ is defective, we can construct each Jordan block in the following way:
a $k$ by $k$ Jordan block $J_{\lambda_i}$ of $\lambda_i$ ($\lambda_i\neq 0$) in $J$ is 
\[J_{\lambda_i} = 
\begin{pmatrix}
\lambda_i & 1 &  &  \\
& \lambda_i & \ddots &  & \\
&  & \ddots &  & 1\\
&  &   &  & \lambda_i\\
\end{pmatrix},
\]
let the corresponding block $J'_{\lambda_i}$ in $J'$ be the inverse of $J_{\lambda_i}$
\[J'_{\lambda_i} \coloneqq J_{\lambda_i}^{-1} = 
\begin{pmatrix}
\frac{1}{\lambda_i} & -\frac{1}{\lambda_i^2} &  &\cdots & (-1)^{k+1}\frac{1}{\lambda_i^k} \\
&  \frac{1}{\lambda_i}  & -\frac{1}{\lambda_i^2}&  \cdots & (-1)^{k}\frac{1}{\lambda_i^{k-1}}\\

&  & \ddots & \ddots & \vdots\\
&  &  &  \frac{1}{\lambda_i}  & -\frac{1}{\lambda_i^2}\\ 
&  &  &  & \frac{1}{\lambda_i}\\
\end{pmatrix}.
\]
For a $k$ by $k$ Jordan block of diagonal zero ($\lambda_i = 0$), which is the nilpotent matrix $N$, we can set the corresponding block in $J'$ as a zero matrix $0_k$ . Since $N$ is not invertible, letting the block be $0_k$ will have the same result as setting it as $N^{k-1}$.
There is a certain freedom in the choice of this block.

Note that, for invertible channels, $\mc{N}^+$ described above provides the inverse $\mc{N}^{-1}$ of the channel ($\mc{N}^+ = \mc{N}^{-1}$).
For non-invertible channels, this construction~\cref{eq:inverse} does not satisfy the condition of generalized inverse ($\mc{N}\circ \mc{N}^{+}\circ\mc{N} \neq \mc{N}$ when the dimension of the nilpotent Jordan block is greater than one). We will call $\mc{N^+}$ the \textit{Drazin-inverse} since it is the same construction as the Drazin inverse in matrix analysis~\cite{drazin1958pseudo}.

The resulting composed map $\mathbf{v}(\mc{N})\mathbf{v}(\mc{N}^+) = QJ''Q^{-1}$, where $J''$ is a diagonal matrix with only 0's and 1's on its main diagonal. When the noise channel $\mc{N}$ already only contains $1$ and $0$ in its spectrum, the Drazin-inverse is itself, and does not recovery more information. In fact, any generalized inverse would not improve the outcome in this case. 

The following proposition tells us another condition for a quantum channel to have a non-CP (Drazin-) inverse.

\begin{prop}\label{thm:non-p} 
	If a non-zero eigenvalue $\lambda$ of a quantum channel $\mc{N}$ has modulus less than $1$ ($|\lambda|<1$), then the inverse (or Drazin-inverse) channel $\mc{N}^{+}$ is not completely positive.
\end{prop}

\begin{proof}
	$\mc{N}$ is a CPTP map, therefore its spectral radius is one~\cite{wolf2010inverse}, i.e. $|J_{ii}| \le 1$ for any main diagonal element $J_{ii}$ in $J$. Since $\mc{N}$ has eigenvalues less than $1$, there exists $|J_{jj}| < 1$ for some $j\in \{1,\cdots,d^2\}$. 
	As defined above, $|J'_{jj}| > 1$, i.e. the spectral radius of $\mc{N}^{+}$ is greater than one. Therefore, $\mc{N}^{+}$ is not complete positive.
\end{proof}

Any general (non-unitary) quantum channel cannot be fully recovered by another channel since CPTP maps cannot resolve contractions in the subspace corresponding to $|\lambda_i|<1$.
To deal with contractions, the CP property has to be broken.
That is to say, restricting the generalized inverse $\mc{N}^{g}$ to be CPTP has to scarifies the quality of recovery.
However, while HPTP maps cannot be directly implemented in a physical system, Steinspring dilation theory guarantees CPTP maps can be extended to unitary channels and hence executed in physical devices. This inspires several beautiful works on finding CPTP inverses that optimize the average fidelity of the recovery~\cite{karimipour2020quasi,shahbeigi2021quasi,aurell2015time}.
On the other hand, HPTP maps can be decomposed as a linear combination of CPTP maps, and so can be implemented in a physical device.
The physical implementability of HPTP maps is defined and discussed in~\cite{jiang2020physical,regula2021operational}.
Also note that the spectrum of a quantum channel can be defined independently from its representations. 
In this section, we mainly work with superoperators (natural representation), but the Proposition~\ref{thm:non-p} still holds in other representations (e.g. the Pauli representation).

Unlike the Choi representation, the natural representation does not directly show a lot of critical properties of quantum channels, such as CP, TP, or HP. 
However, we found that the eigen-structure of the superoperator is essential for its property. 
\cref{lemma1} and \cref{lemma2} provide an insight into why Moore-Penrose inverse is not TP in certain cases.
Then, we prove that the Drazin-inverse for a TP map is always TP in~\cref{thm:tp}.

Denote the trace operation in the vector representation $\Sp{A}$ of a $d$ by $d$ matrix $A$ as $\str{\cdot}$, where $\str{\Sp{A}} \coloneqq \Tr(A)$.

\begin{lemma}\label{lemma1}
	If a linear map $\mc{N}: M_d \to M_d$ is trace preserving, the eigenvectors $v$ and generalized eigenvectors $v^g$ of eigenvalue $\lambda\neq 1$ of the superoperator $\mathbf{v}(\mc{N})$ is trace zero, i.e. $\str{v} = \str{v^g} = 0$.
\end{lemma}
\begin{proof}
	For an eigenvector $v$ of $\Sp{\mc{N}}$, we have $\mathbf{v}(\mc{N})v = \lambda v$. Since $\mc{N}$ is trace preserving, $\str{v} = \str{\lambda v}$. And the eigenvalue $\lambda\neq 1$, we have $\str{v} = 0$
	
	For a $k$ by $k$ Jordan block of eigenvalue $\lambda^g$, where $k>1$, denote the first generalized eigenvector as $v^{g_1}$, we have 
	\begin{equation}\label{eq:gen_eig}
	[\mathbf{v}(\mc{N}) - \lambda^g I]v^{g_1} = v,
	\end{equation}
	where $v$ is the eigenvector corresponding to $\lambda^g$. 
	Taking the trace on both sizes, $\str{(\mathbf{v}(\mc{N}) - \lambda^g I)v^{g_1}} = \str{v}$, the left hand side is $\str{v^{g_1}-\lambda^gv^{g_1}} = (1-\lambda^g)\str{v^{g_1}}$, and the right hand side is zero from the argument above.
	Since $\lambda^g \neq 1$, $\str{v^{g_1}} = 0$.
	By deduction, all $v^{g_i}$ are trace zero for $i\in\{1,\cdots,k-1\}$.
\end{proof}

\begin{lemma}\label{lemma2}
	For a trace persevering linear map $\mc{N}: M_d \to M_d$, if there is a $k$ by $k$ ($k>1$) defective Jordan Block of eigenvalue $\lambda = 1$ in $\Sp{\mc{N}}$, the eigenvector $v$ and first $k-2$ generalized eigenvector $v^{g_{i}}$ has to be trace zero, i.e. $\str{v} = \str{v^{g_{i}}} = 0$ for $i\in\{1,\cdots,k-2\}$.
\end{lemma}
\begin{proof}
	Assume that $\str{v}\neq 0$. The first generalized eigenvector $v^g$ satisfy that $[\mathbf{v}(\mc{N}) - I]v^{g_1} = v$.
	Taking trace on both size, the left hand side equals to zero, and the right hand side does not equal to zero, leading to a contradiction. 
	The same argument holds for the rest of the generalized eigenvectors except the last one.
\end{proof}

From~\cref{lemma1} and~\cref{lemma2}, we know that all eigenvectors $v_\lambda$ for $\lambda\neq 1$ of a TP map has to be traceless. When $\lambda = 1$, if its algebraic multiplicity equals to its geometry multiplicity, $\str{\Sp{\mc{N}}v_\lambda} = \str{v_\lambda}$ (i.e. the trace of $v_\lambda$ will not be changed under the action of $\Sp{\mc{N}}$); if the algebraic multiplicity does not equal to the geometry multiplicity, the eigenvectors and generalized eigenvectors is traceless except for the last generalized eigenvector.
This tells us that the eigen-structure of the superoperator $\Sp{\mc{N}}$ is crucial for $\mc{N}$ to be TP. The way that we construct the Drazin-inverse $\mc{N}^+$ largely preserves the eigen-structure, while the Moore-Penrose inverse $\mc{N}^p$ focuses more on the singular value structure. It hints that $\mc{N}^+$ should be TP and $\mc{N}^p$ may not.

\begin{theorem}\label{thm:tp}
The Drazin-inverse $\mc{N}^+$ of a trace preserving map $\mc{N}$ is also trace preserving.
\end{theorem}
To prove that $\mc{N}^+$ is trace preserving,
we need to prove
\[
\str{\mathbf{v}(\mc{N}^+)v_\lambda} = \str{v_\lambda},
\]
for every eigenvectors and generlized eigenvectors $v_\lambda$ of $\Sp{\mc{N}}$ in $Q$. From the construction of $\mc{N}^+$, we almost get trace preserving for free. The proof can be found in~\cref{app:tp}. Moreover, it is easy to see from the proof that the composed map $\mc{N}^+\circ\mc{N}$ is also trace preserving.

\begin{example}%\label{exp:inv}
Here we give an example where the Moore-Penrose inverse $\mc{N}^p$ of a CPTP map is not TP, while the Drazin-inverse $\mc{N}^+$ is TP.
Consider a noise channel $\mc{N}$ whose Choi representation is given by
	\begin{equation}\label{eqn_ex2_noise}
	C(\mc{N})=\frac{1}{20}\left(\begin{array}{cc|cc}
	8 & 0 & 1 & 6 \\
	0 & 12 & 2 & -1 \\
	\hline 1 & 2 & 8 & 0 \\
	6 & -1 & 0 & 12
	\end{array}\right),
	\end{equation}
	and its superoperator is 
	\begin{equation*}
	\mathbf{v}(\mc{N})=\frac{1}{20}\left(\begin{array}{cccc}
	8 & 1 & 1 & 8 \\
	0 & 6 & 2 & 0 \\
	0 & 2 & 6 & 0 \\
	12 & -1 & -1 & 12
	\end{array}\right)
	\end{equation*}
	
	The Jordan normal form is given by $J = \text{diag}(0,1,\frac25,\frac15)$, and its inverse is $J' = \text{diag}(0,1,\frac52,5)$.
	
	The superoperator of Drazin-inverse $\mc{N}^{+}$ is 
	\begin{equation*}
	\mathbf{v}(\mc{N}^{+})=\left(\begin{array}{cccc}
	\frac25 & \frac5{16} & \frac5{16} & \frac25 \\
	0 & \frac{15}4 & -\frac54 & 0 \\
	0 & -\frac54 & \frac{15}4 & 0 \\
	\frac3{5} & -\frac5{16} & -\frac5{16} & \frac3{5}
	\end{array}\right)
	\end{equation*}
	
	The Choi representation of $\mc{N}^{+}$ is 
	\begin{equation*}
	C(\mc{N}^{+})=\left(\begin{array}{cc|cc}
	\frac25 & 0 & \frac5{16} & \frac{15}{4} \\
	0 & \frac3{5} & -\frac5{4} & -\frac5{16} \\
	\hline \frac5{16} & -\frac5{4} & \frac25 & 0 \\
	\frac{15}{4} & -\frac5{16} & 0 & \frac3{5}
	\end{array}\right)
	\end{equation*}
	
	%The eigenvalues of $C(\mc{N}^+)$ are $[-3.28983007, -0.7155248 ,  1.7155248 ,  4.28983007]$.
	The Choi representation has negative eigenvalues.
	Therefore, the channel $\mc{N}^+$ is trace preserving (partial trace of $C(\mc{N}^{+})$ is identity)%\todo{(junan: how can we see from this that $\mc{N}^+$ is TP?)}
	, Hermitian preserving ($C(\mc{N}^{+})$ is Hermitian), but not complete positive.
	
	The Moore-Penrose inverse of $\Sp{\mc{N}}$ is
	\[
	\Sp{\mc{N}^{p}} =
	\begin{pmatrix}
	\frac{115}{294} & \frac{10}{441} & \frac{10}{441} & \frac{505}{882} \\
	\frac{50}{147} & \frac{3245}{882} & -\frac{1165}{882} & -\frac{100}{441} \\
	\frac{50}{147} & -\frac{1165}{882} & \frac{3245}{882} & -\frac{100}{441} \\
	\frac{115}{294} & \frac{10}{441} & \frac{10}{441} & \frac{505}{882}
	\end{pmatrix},
	\]
	and its Choi representation is
	\[
	C(\mc{N}^{p}) = 
	\left(\begin{array}{rr|rr}
	\frac{115}{294} & \frac{50}{147} & \frac{10}{441} & \frac{3245}{882} \\
	\frac{50}{147} & \frac{115}{294} & -\frac{1165}{882} & \frac{10}{441} \\
	\hline \frac{10}{441} & -\frac{1165}{882} & \frac{505}{882} & -\frac{100}{441} \\
	\frac{3245}{882} & \frac{10}{441} & -\frac{100}{441} & \frac{505}{882}
	\end{array}\right),
	\]
	which is Hermitian preserving but not trace preserving.
	\label{exp:pseudo}

	\end{example}

To our knowledge, generalized non-CPTP inverses for non-invertible quantum channels have not been studied extensively in previous literature. 
The new understanding of the natural representation opens the possibility of studying the structures and properties of these maps from mathematical interests. 
Different constructions of generalized inverse maps bring new properties.
This research direction can provide a guideline for implementing these maps in EM and also in other areas of quantum information sciences.

\section{QEM in Quantum Computation}\label{sec_conseq}
In the previous section we have discussed how the nature of noise determines whether it is theoretically possible to fully recover the quantum and/or classical information, under the framework of classical and quantum communication. 
We explored the possibility of non-invertible noise and constructed a Drazin-inverse under such a case.
In this section, we will study the effects of recovery operations when performing QEM in quantum computation.
The task of quantum computation may be viewed as a modified version of communication, where Alice and Bob are no longer spatially separated, but the channel from Alice (the input) to Bob (the output) becomes a nontrivial unitary $\mc{U}$.

Using the intuition built from the communication setting, there can again be two situations depending on the goal.
In the first case, one is interested in obtaining the expectation value of some observable $\Tr[O \rho_{\text{out}}]$ from the output state.
We can call this ``quantum computation of classical information'', inspired by the usage of a similar term when discussing about quantum communication.
Accurately obtaining $\Tr[O \rho_{\text{out}}]$ for some observable $O$ is the goal of many quantum algorithms, such as the Harrow-Hassidim-Lloyd algorithm~\cite{Harrow2009} for solving linear systems and the variational quantum eigensolver algorithm~\cite{Peruzzo2014a} for estimating ground/excited state energies of molecules.
From our discussion before, it is now clear that both QEC and QEM can be useful for these types of problems.
Indeed, to the authors' knowledge, all preexisting QEM protocols have been developed towards solving these problems.

In the second case, the goal is to output a particular quantum state, which needs to be kept in coherence and perhaps sent to another party later.
This is particularly relevant in, for example, the problem of distributed quantum computing~\cite{Kimble2008,Wehner2018} or active quantum memories~\cite{Bacon2006,Terhal2015}.
This is thus the case of ``quantum computation of quantum information'', with a reasoning similar to the previous one.
Here, the idea of QEM cannot be used (at least locally at the output end), to reduce noise effects for these types of tasks, and only QEC is useful.

As mentioned previously, quantum computation has an additional layer of complexity that comes from the composition of gates.
In communication, it is known prior to sending particles that the ideal ``gate'' is the identity, so applying the inverse noise map would directly yield the ideal input distribution.
In computation, typically only a condensed description of the target unitary in the form of a quantum circuit is available prior to an experiment.
The full form of the target unitary map (and consequently, that of the noise process) is typically not known explicitly.
Instead, one usually have knowledge on components (e.g., the form of one- or two-qubit gates) in the circuit as well as the noise on these components.
Thus, QEM protocols reviewed in \cref{sec_review_QEM} were developed to implement this noise inversion procedure more efficiently in practice, for ``quantum computation of classical information'' tasks.
In the following we study this scenario in more detail, with a focus on the effects of imperfect knowledge on noise in \cref{sec_imperfect_noise}.

\subsection{QEM in Multi-layer Quantum Computation}\label{sec_multi_gate}
Consider again a noisy quantum circuit with depth $n$, where each layer can be represented by a unitary map $\mc{U}_{i}$ with $i=1,...,n$. The ideal output would be
\[
\rho_\text{out}^\text{ideal} = \mathcal{U}_n\circ\cdots\circ\mathcal{U}_1(\rho_\text{in}).
\]

In practice the gates $\mc{U}_{i}$ are implemented imperfectly.
Making the standard Markovian assumption on the noise, each imperfect $\mc{U}_{i}$ can be decomposed as $\mc{N}_{i} \mc{U}_{i}$, where each $\mc{N}_{i}$ is a CPTP map and can be distinct for different $i$. We thus have
\begin{equation}\label{eq:markovian}
\rho_{\text{out}}^{\text{exp}} = \mathcal{N}_n\circ \mathcal{U}_n\circ\cdots\circ\mathcal{N}_1\circ \mathcal{U}_1(\rho_\text{in})
\end{equation}
where $\rho_\text{in}$ is the input quantum state, $\rho_\text{out}$ is the quantum state came out of the noisy circuits, $\mathcal{U}_i$ are the desired operations, and $\mathcal{N}_i$ are the noise channels corresponding to gate $\mc{U}_{i}$.

To perform QEM, one first tries to learn (part or all of) the noise models, then recover the ideal gates through either physical or numerical means.
Thus, if we wish to analyze the performance of the \emph{best possible} QEM strategy, we may wish that all $\mathcal{N}_i$'s are known exactly.
But in reality, these $\mathcal{N}_i$'s are obtained from experiments either during the calibration stage or as part of the QEM process, which necessarily involves inaccuracies when being reconstructed.
Denote the experimentally characterized noise models $\tilde{\mathcal{N}}_i$, and let $\tilde{\mathcal{N}}_i^{-1}$ denote the inverse of $\tilde{\mathcal{N}}_i$.
In this section we will consider channels with the same input and output dimensions.

First, consider the case where $\mathcal{N}_i^{-1}$ exists and is CPTP for all $i$.
Recall that this is true \emph{iff} $\mathcal{N}_i$ is a unitary channel when the input and output dimensions equal.
Then in principle one can insert an additional gate implementing $\mathcal{N}_i^{-1}$ after each $\mc{U}_{i}$ to fully invert the noise effect~\cite{jiang2020physical}.
In reality, the experimentally obtained noise models are $\tilde{\mathcal{N}}_i$.
Thus, the output from this method will be
\begin{equation}\label{eq:rhoem_CPTP}
\rho_\text{EM} = \tilde{\mathcal{N}}_n^{-1}\circ\mathcal{N}_n\circ \mathcal{U}_n\circ\cdots\circ\tilde{\mathcal{N}}_1^{-1}\circ\mathcal{N}_1\circ \mathcal{U}_1(\rho_\text{in}).
\end{equation}

Naturally, there are two main sources of additional errors. 
First, the experimentally learned noise model $\tilde{\mathcal{N}}_i$ is not always equal to $\mathcal{N}_i$, so $\tilde{\mathcal{N}}_i \circ \mc{N}_{i}$ is not necessarily equal to the identity.
Second, even if $\mathcal{N}_i$ can be learned ideally, physically implementing $\mc{N}_{i}^{-1}$ will also not be ideal and can introduce extra errors.

Next, consider the case where $\mathcal{N}_i^{-1}$ exists but is not CPTP.
In this case it is impossible to physically restore the ideal output state.
However, we can still perform the inverse numerically to recover the output density matrix.
This can be thought of as first numerically inverting all the channels in \cref{eq:markovian}, then applying the ideal gates in the original order.
Specifically, we define the ideal ``reversal'' channel $\mc{R}$ to be the one that maps $\rho_{\text{out}}^{\text{exp}}$ to $\rho_{\text{in}}$, constructed as
\begin{equation}
	\mc{R} \coloneqq \mc{U}_1^\dag\circ\mc{N}_1^{-1}\circ \cdots \circ \mc{U}_n^\dag\circ\mc{N}_{n}^{-1}.
\end{equation}
Correspondingly, replacing $\mc{N}_{i}$ in the above by $\mc{\tilde{N}}_{i}$ gives the realistic reversal channel,
\begin{equation}
	\tilde{\mc{R}} \coloneqq \mc{U}_1^\dag\circ\mc{\tilde{N}}_1^{-1}\circ \cdots \circ \mc{U}_n^\dag\circ\tilde{\mc{N}}_{n}^{-1},
\end{equation}
which represents the experimentalist's best knowledge about the ideal reversal channel $\mc{R}$.
We thus have
\begin{align}\label{eq:nlayer}
	\rho_\text{EM} = \mc{U}_{n\cdots 1} \circ \tilde{\mc{R}} (\rho_\text{out}^\text{exp})
\end{align}
where the shorthand $\mc{U}_{n\cdots 1}: = \mc{U}_n\circ \cdots \circ \mc{U}_1$ is used for the ideal circuit sequence.
The composition of such channels first maps the experimental output state $\rho_\text{out}^\text{exp}$ back to the input state $\rho_\text{in}$, 
then perform the ideal operations $\mc{U}_{n\cdots 1}$; this is illustrated by the blue arrows in~\cref{fig:maps}.
The numerical inverse method does not involve implementing physical gates, but still require that noise processes are accurately characterized.
As reviewed in \cref{sec_review_QEM}, many current numerical QEM protocols can be categorized as trying to obtain the exact noise-inverted output, meaning that their optimal performance is upper bounded by \cref{eq:nlayer}.

A naive numerical implementation of the channel inverse requires simulating the quantum circuit $\mc{U}_{n\cdots 1}$, which is naturally expensive.
Generally speaking, the computational complexity for computing \cref{eq:nlayer} can be higher than classically simulating the ideal circuit, even without including the cost of characterizing noise channels. 
Therefore, directly computing such an inverse channel is not efficient for mitigating error in practice. 
However, comparing the error mitigated results with classically simulated ones may reveal how precise our knowledge about the device noise is.
Moreover, the result from this ``optimal method'' upper bounds the performance of any error mitigation protocol.

Finally, we mention briefly that if only an approximate version of the ideal output is wanted, it may be sufficient to apply one effective recovery map $\mathcal{N}_\text{eff}^{-1}$ to the noisy output state, in the hope that it will eliminate most of the noise effects.
The mitigated output from this approximate method is given by
\begin{equation}
\rho_\text{EM} = \mathcal{N}_\text{eff}^{-1}(\rho_\text{out}^\text{exp}).
\end{equation}
The effective recover map $\mathcal{N}_\text{eff}^{-1}$ may be applied either physically or numerically, depending on the situation.
Typically, $\mathcal{N}_\text{eff}^{-1}$ contains a few tunable parameters which can be experimentally optimized to achieve the best noise-mitigating performance on some test runs.
Methods that fall into this category include decoherence compensation in NMR experiments, and depolarizing-model-based EM~\cite{vovrosh2021efficient}.
Recent work also considered continuous inversion through the Petz recovery map~\cite{kwon2021reversing}.
Although it leads to resource savings in practice, this general approach can be problematic in some instances due to channel mismatching. We provide one such example in~\cref{exp:mismatch}.

\subsection{QEM with Imperfect Knowledge on Noise}\label{sec_imperfect_noise}
Next we study the effects of imperfectly characterized noise channels on the performance of QEM.
It is generally acknowledged that characterizing noise models in a quantum system is highly resource demanding~\cite{nielsen2021gate}.
In many current error mitigation protocols, the noise channel is assumed to be particular models~\cite{2021simplemitigation}, such as a depolarizing channel $\mc{D}$. 
It is then natural to ask the question of how incorrectly characterized noise channels $\{\mc{N}_i\}$ would affect the mitigation outcome.
As mentioned before, we will assume all $\{\mc{N}_i\}$'s to be invertible in this subsection. 

\cref{fig:maps} shows an illustrative diagram showing the relationship between different objects that we will discuss about.
One sees that while the ideal circuits and the experimental operations are CPTP maps, the channels $\mc{R}$ and $\tilde{\mc{R}}$ are not necessarily CPTP anymore.
The gap $\tN_i - \mc{N}_i$ between the estimations $\tN_i$ and the actual channels $\mc{N}_i$ upper bounds the result of EM, independent of how the inverses are achieved. 
And this gap only affects the difference between $\mc{R}$ and $\tilde{\mc{R}}$.

\begin{figure}[ht]
	\centering
	\includegraphics[width = \linewidth]{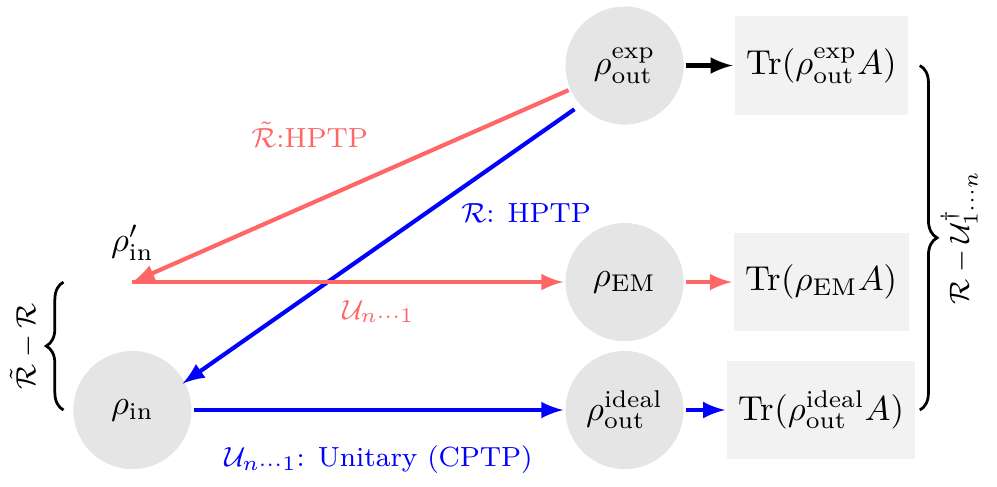}
	\caption{The schematic diagram of maps. 
	The blue arrows indicate the map $\mc{U}_{n\cdots 1}\circ \mc{R}$ for ideal error mitigation, and the red arrows indicate the map $\mc{U}_{n\cdots 1}\circ \tilde{\mc{R}}$ for error mitigation with imperfect noise characterization. 
	The error between actual noise channels $\mc{N}_i$ and estimations $\tN_i$ cause the difference between $\mc{R}$ and $\tilde{\mc{R}}$, which leads to a deviation in the mitigated result.
	}\label{fig:maps}
\end{figure}

From the perspective of output states, the goal of EM is to bring the output states closer to the ideal. 
In terms of state fidelity, this is to ensure that 
\begin{equation}\label{eq:fid_state}
F(\rho_\text{EM},\rho^\text{ideal}_\text{out})> F(\rho^\text{exp}_\text{out},\rho^\text{ideal}_\text{out}),
\end{equation}
where $F(\rho_1,\rho_2) \coloneqq \tr(\sqrt{\sqrt{\rho_1}\rho_2\sqrt{\rho_1}})$ is the fidelity between $\rho_1$ and $\rho_2$.

If the actual noise channels $\{\mc{N}_i\}$ are invertible and the noise characterization is perfect ($\tN_i = \mc{N}_i$), theoretically the errors can be perfectly mitigated, with \cref{eq:fid_state} naturally satisfied. 
Realistically, $\tN_i \neq \mc{N}_i$, which opens the gap between ideal output states $\rho^\text{ideal}_\text{out}$ and error mitigated state $\rho^\text{EM}$. 
We  next answer the question of how much will the imperfections in characterizing $\mc{N}$ worsen the fidelity.

Let $\Delta \mc{N}_i \coloneqq \tN_i - \mc{N}_i$ and $\Delta \mc{N}_i^{-1} \coloneqq \tN_i^{-1} - \mc{N}_i^{-1}$.
Note that $\Delta \mc{N}_i $ and $\Delta \mc{N}_i^{-1}$ are related by $\Delta \mc{N}_i \tN_i^{-1} + \mc{N}_i\Delta \mc{N}^{-1}_i = 0$.
We mainly use $\Delta \mc{N}_i^{-1}$ in later discussion.

\cref{fig:maps} shows that the errors $\{\Delta \mc{N}_i^{-1}\}$ only affect $\mc{R}$ and $\tilde{\mc{R}}$ in the error mitigation maps.
The difference between $\rho_\text{EM}$ and $\rho_\text{out}^\text{ideal}$ is
\begin{align}\label{eq:error}
\rho_\text{EM} - \rho_\text{out}^\text{ideal} & = \mc{U}_{n\cdots 1}\circ \tilde{\mc{R}}(\rho_\text{out}^\text{exp}) - \mc{U}_{1\cdots n}(\rho_\text{in}) \nonumber\\
%&= \mc{E}^{-1}_\text{EM}(\rho_\text{out}^\text{exp}) - \mc{E}^{-1}_\text{EM-ideal}(\rho_\text{out}^\text{exp}) \nonumber\\
& = \mc{U}_{n\cdots 1}\circ\left[\tilde{\mc{R}} - \mc{R} \right](\rho_\text{out}^\text{exp})\\
& \coloneqq \mc{U}_{n\cdots 1}\circ \Delta \mc{N} (\rho_\text{out}^\text{exp}).
\end{align}

In the middle bracket in~\cref{eq:error}, the errors $\{\Delta \mc{N}_i^{-1}\}$ scramble in the layers of unitaries $\mc{U}_i^\dag$. 
Denote the first order estimation of $\Delta\mc{N}$ to be $\Delta \mathcal{N}^{(1)},$ where each term in $\Delta \mathcal{N}^{(1)}$ only contain one of $\Delta \mc{N}_i^{-1}$ (see~\cref{eq:delta_rho} in~\cref{app:imperfect} for the explicit expression).
The first order error between states is $\Delta \rho_\text{EM} \coloneqq \mc{U}_{1\cdots n}\circ \Delta\mc{N}^{(1)}(\rho_\text{out}^\text{exp})$. 
We then define $F(\rho_\text{EM}, \rho_\text{EM}+\Delta\rho_\text{EM})$ to be the first order estimation $F^{(1)}(\rho_\text{EM},\rho_\text{out}^\text{ideal})$ of the state fidelity $F(\rho_\text{EM},\rho_\text{out}^\text{ideal})$. %Let $ \coloneqq F(\rho_\text{EM}, \rho_\text{EM}+\Delta\rho_\text{EM})$.
The following proposition gives a bound on this quantity.

\begin{prop}\label{prop:fid}
	The first order estimation of fidelity between $\rho_\text{EM}$ and $\rho_\text{out}^\text{ideal}$ is %\jl{we need to define what first order means}
	\begin{align}\label{eq:fid_bound}
	& \left(1 - \frac{1}{2}\sqrt{d} C_\text{exp} \left\|\mathbf{v}(\Delta \mc{N}^{(1)})\right\|\right)^2 \le F^{(1)}(\rho_\text{EM},\rho_\text{out}^\text{ideal}) \nonumber \\
	&\le 1-\frac14 \left(l_U\cdot \left\|\mathbf{v}(\Delta \mc{N}^{(1)})\Sp{\rho_\text{out}^\text{exp}}\right\|\right)^2,
	\end{align}
	where $C_\text{exp}\coloneqq \|\Sp{\mc{U}_{n\cdots 1}}\|\cdot\|\Sp{\rho_\text{out}^\text{exp}}\|$ is an experiment-related constant, %and $\left\|\Sp{\Delta \mathcal{N}^{(1)}}\right\|$.
	and $l_U \coloneqq \inf_{\|x\| = 1} \|\mathbf{v}(\mc{U}_{n\cdots1}) x\|$ is the lower Lipschitz constant  of the ideal operations $\mc{U}_{n\cdots 1}$.
	The norm $\|\cdot\|$ is 2-norm for vectors and is the induced matrix norm for matrices.
\end{prop}

We can see that $F^{(1)}(\rho_\text{EM},\rho_\text{out}^\text{ideal})$ is bounded by $\Delta \mathcal{N}^{(1)}$ and experimental constants (including norms of the ideal circuits and the Frobenius norm of the experimental outcome state $\rho^\text{exp}_\text{out}$).
Therefore, by bounding the errors $\{\Delta \mc{N}^{-1}_i\}$ in channel estimation, one can constrain the fidelity by using~\cref{eq:fid_bound}.
In fact, this result can be understood easily from the left-hand side of~\cref{fig:maps} -- closing the gap between $\tilde{\mc{R}}$ and $\mc{R}$ can bring $\rho_\text{in}'$ and $\rho_\text{in}$ closer, therefore bounding the fidelity afterwards.
Further details can be found in~\cref{app:imperfect}.

\begin{figure}[ht]
	\centering
	\includegraphics[width = .6\linewidth]{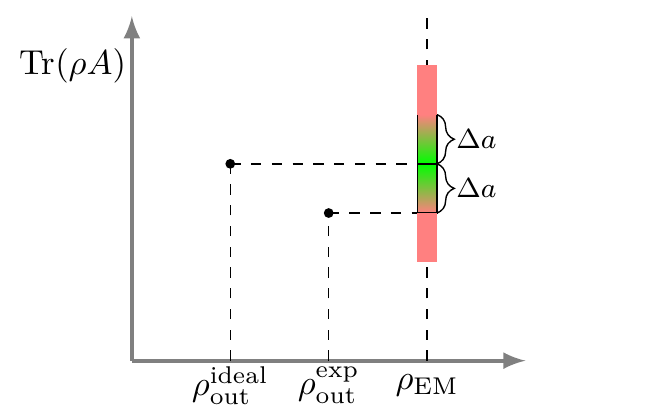}
	\caption{A schematic diagram for improving the expectation value $\Tr(\rho A)$ using QEM. 
	$\Delta a$ denotes the difference (in absolute value) between the ideal output, $\rho_\text{out}^\text{ideal}$, and the actual output $\rho_\text{out}^\text{exp}$. 
	The goal of QEM is to achieve a mitigated output state $\rho_\text{EM}$, such that $\Tr(\rho_\text{EM} A)$ in the green zone for an observable $A$ of interest.}\label{fig:sch}
\end{figure}

If the task realized by the given circuit only concerns the expectation value of a set of observables $\{A_i\}$, then the goal of QEM can be simplified as recovering the ideal expectation value, $\Tr(\rho_{\text{out}}^{\text{ideal}} A_i)$.
As shown in~\cref{fig:sch}, one would like the error mitigated result to be closer to the ideal than the one directly coming from experiments. 
Since one cannot perfectly characterize the noise models $\mc{N}_i$, it is desirable to know the condition which guarantees $\Tr(\rho_\text{EM} A)$ to land in the green zone.
We show in~\cref{app:exp_val} that the following is a sufficient condition for such a goal. 
\begin{prop}\label{prop:exp_val}
	If the following condition~\cref{eq:suff} is satisfied, 
	Quantum Error Mitigation has the ability
	to improve the expectation value of any observable $A$ for any circuit $\mc{U}_{n\cdots 1}$.
	\begin{equation}\label{eq:suff}
	\left\|\Sp{\Delta\mc{N}}\right\| \le l_\text{ideal-exp},
	\end{equation}
	where $l_\text{ideal-exp}\coloneqq \inf_{\norm{x} = 1}\norm{ \Sp{\mc{R} - \mc{U}^\dag_{1\cdots n}}x }$ is the lower Lipschitz constant of $\Sp{\mc{R} - \mc{U}^\dag_{1\cdots n}}$.
\end{prop}
In the above result, the channels $\mc{R}$ and $\mc{U}^\dag_{1\cdots n}$ maps $\rho_\text{out}^\text{exp}$ and $\rho_\text{out}^\text{ideal}$ back to $\rho_\text{in}$ respectively.
The condition~\cref{eq:suff}, in general, is requiring $\Delta\mc{N}$ to be smaller than $\mc{R} - \mc{U}^\dag_{1\cdots n}$.
It is straightforward to observe from the brackets in~\cref{fig:maps}.
Since this proposition is for any observables and any circuit, it will also work for quantum state fidelity.

Note that \cref{eq:suff} is a stringent requirement. 
If $\Sp{\mc{R} - \mc{U}^\dag_{1\cdots n}}$ has a nontrivial null space, then it will force the noise channel estimation $\tN_i$ to be perfect, i.e. $\tN_i = \mc{N}_i$ for $\forall i\in \{1,\cdots,n\}$. 
We do not make extra assumptions on circuits and noises while deriving this sufficient condition.
Knowing more information about the circuit and noises can loosen the requirement.

Roughly speaking, for an error mitigation protocol to improve the experimental outcome, our knowledge of the noise channels needs to be more accurate than the quality of the experiment.
Since the unitary $\mc{U}_{n\cdots 1}$ is an isometry under Frobenius norm ($\|M\|_F \coloneqq \sqrt{\sum_{ij}|M_{ij}|^2}$), the difference between error mitigated state $\rho_\text{EM}$ and ideal output state $\rho_\text{out}^\text{ideal}$ equals to $\Delta_\rho \coloneqq \|\rho_\text{EM} - \rho_\text{out}^\text{ideal}\|_F  = \|\Delta\mc{N}(\rho^\text{exp}_\text{out})\|_F = \|\tilde{\mc{R}}(\rho_\text{out}^\text{exp}) - \rho_\text{in}\|_F$.
The error $\Delta$ of error mitigation is bounded by $\Delta_\rho$:
\begin{align}
\Delta & = \tr[A(\rho_\text{EM} - \rho_\text{out}^\text{ideal})] \nonumber\\
&\le \|A\|_F\cdot\|\rho_\text{EM} - \rho_\text{out}^\text{ideal}\|_F \nonumber\\
%& = \|A\|_F\cdot \|\tilde{\mc{R}}(\rho_\text{out}^\text{exp}) - \rho_\text{in}\|_F \\ 
& = \|A\|_F\cdot\Delta_\rho.\label{eq:delta}
\end{align}
Assume that the input state $\rho_\text{in}$ is the all-zero state $\ketbra{0,\cdots,0}{0,\cdots,0}$, i.e. $\rho_\text{in}$ is a very sparse matrix with 1 as the first entry and 0 elsewhere.
Note that $\tilde{\mc{R}}$ is noisy thus can be efficiency computed using tensor network methods, and $\rho_\text{out}^\text{ideal}$ can be constructed by classical shadow tomography~\cite{huang2020predicting}. 
The upper bound~\cref{eq:delta} of value $\Delta$ can be efficiently computed.

Finally, we discuss potential consequences of incorrect assumptions about the actual noise.
Normally, certain noise models are assumed while identifying device noise.
The assumptions made on noise models lead to savings in parameters and resources in characterization. 
However, the distance between the actual noise $\mc{N}$ in the system and the model assumed will not be arbitrarily close, which opens a gap between the ideal outcomes and error mitigated outcomes of the given circuit.
In particular, if the error model is overly simplified, it can cause problems on EM performance.

We consider the following simple example of a depth-$1$, single qubit quantum channel, where the actual noise $\mc{N}$ is a Pauli Channel, but a depolarizing channel is assumed when mitigating error. 
\begin{example}\label{exp:mismatch}
Suppose one believes that the noise in the system is mainly depolarizing, and tries to use the depolarizing channel $\mc{D}$ to approximate the actual noise.
After fitting the parameters in $\mc{D}$, the inverse $\mc{D}^{-1}$ of the estimated $\mc{D}$ is used to recover information (i.e. $\mc{D}^{-1}\circ\mc{N}(\rho)$).

The Kraus representation of $\mc{N}$ and $\mc{D}$ are $\mc{N} : \{\sqrt{p_1}I, \sqrt{p_2} X, \sqrt{p_3} Y, \sqrt{(1-p_1-p_2-p_3)}Z\}$ and $\mc{D} : \{\sqrt{1-\frac{3\lambda}{4}}I, \sqrt{\frac{\lambda}{4}} X, \sqrt{\frac{\lambda}{4}} Y, \sqrt{\frac{\lambda}{4}}Z \}$.
For a given set of $\{p_1,p_2,p_3\}$, the optimal $\lambda$ to minimize $\|\mc{N}-\mc{D}\|_?$ varies according to different representations and different choices of norm, $\|\cdot\|_?$. 
The symmetry on the parameters in $\mc{D}$ makes it impossible to perfectly capture the noise $\mc{N}$ for $p_i$'s that do not have such a symmetry.

Note that the two vectors, $\vec n \coloneqq (\sqrt{p_1}, \sqrt{p_2}, \sqrt{p_3}, \sqrt{(1-p_1-p_2-p_3)})$ and $\vec d \coloneqq(\sqrt{1-\frac{3\lambda}{4}}, \sqrt{\frac{\lambda}{4}}, \sqrt{\frac{\lambda}{4}}, \sqrt{\frac{\lambda}{4}})$, are also representations for $\mc{N}$ and $\mc{D}$ respectively.
Since $\vec n$ and $\vec d$ are normalized, minimizing the distance between $\mc{N}$ and $\mc{D}$ is equivalent to maximizing $\vec n \cdot \vec d$. i.e. 
\begin{align*}
\max_{\lambda\in [0,1]} & \left\{\sqrt{p_1 (1-\frac{3\lambda}{4})}\right. \\
& \left. +[\sqrt{p_2}+\sqrt{p_3}+\sqrt{(1-p_1-p_2-p_3)}]\sqrt{\frac{\lambda}{4}}\right\}.
\end{align*}

When $p_1=\frac{1}{2}$ and $p_2=p_3=0$, a channel will have a phase flip error with probability $\frac{1}{2}$ and will stay unchanged with probability $\frac{1}{2}$, corresponding to an optimal $\lambda_\text{max}$ value of $\frac13$. 
This $\lambda_\text{max}$ bounds the distance between $\mc{N}$ and $\mc{D}$ from above for this metric. 
Assume one fits the parameter $\lambda$ from experiments, and obtains the estimation that $\lambda = \frac13$, the channel $\{\sqrt{\frac{3}{4}}I, \sqrt{\frac{1}{12}} X, \sqrt{\frac{1}{12}} Y, \sqrt{\frac{1}{12}}Z \}$ will be believed to be $\tN$.
Then $\tN^{-1} = \mc{D}^{-1}$ will be used to perform error mitigation. 
In~\cref{fig:app_e_val}, we can see that while the actual channel $\mc{N}$ preserves the expectation value of $Z$, the mitigated results are actually worse due to the incorrect assumptions on noise model (see the blue triangles in~\cref{fig:app_e_val}).
Also note that, since $\mc{D}^{-1}$ is non-CP, the outputs $\mc{D}^{-1}\circ\mc{N}(\rho)$ are not valid quantum states anymore. 
In this case the fidelity function is not bounded below 1, thus is no longer a valid metric.
We give further details in~\cref{app:depolarizing}.

\begin{figure}[ht]
\includegraphics[width = .9\linewidth]{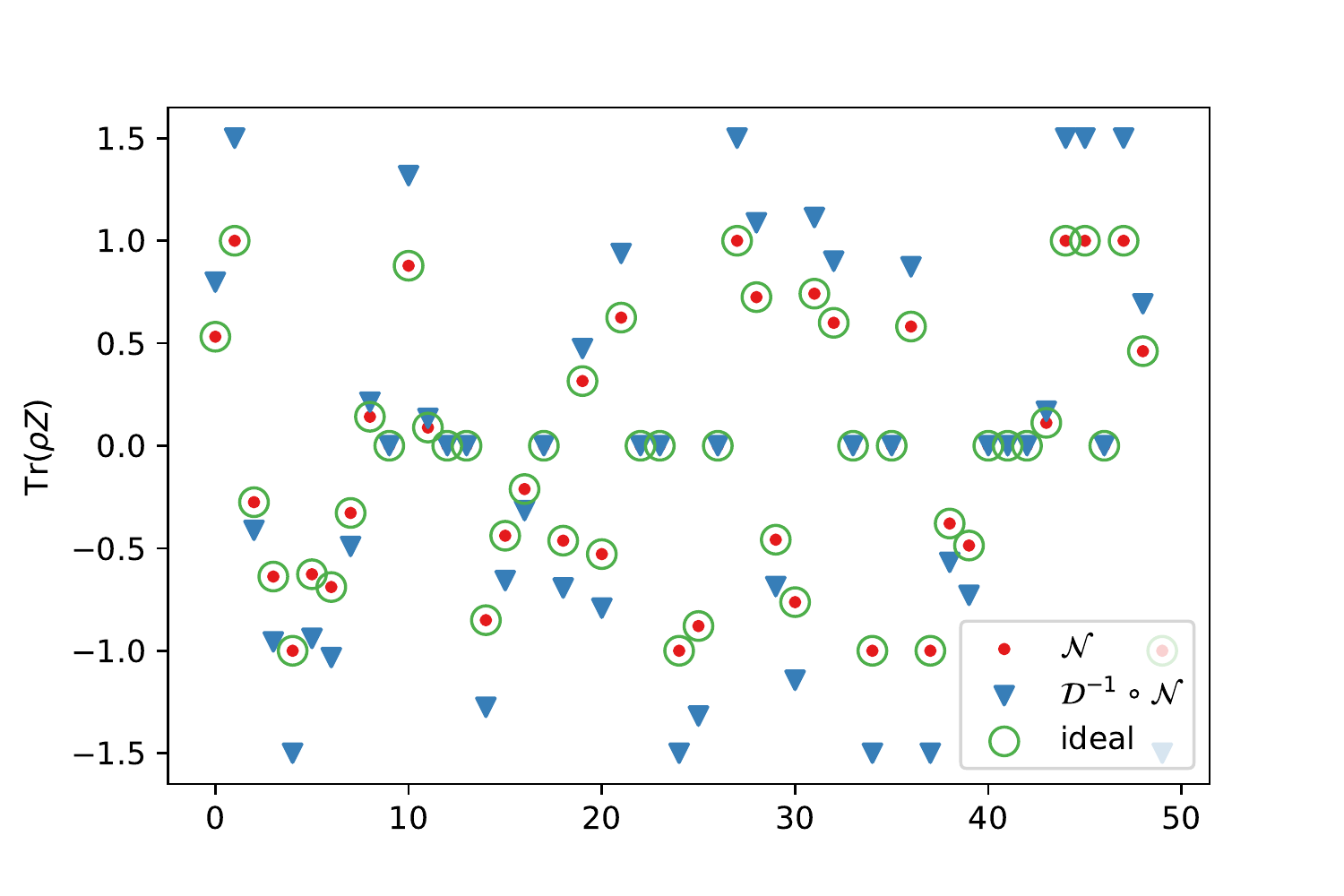}
\caption{Expectation value of $Z$ for 50 randomly generated states. The $x$-axis is the dummy label for these tested states.
}\label{fig:app_e_val}
\end{figure}

\end{example}

While the above is a rather extreme example of channel mismatching, the message in this example is alerting, because it illustrates how a misunderstanding of noise can lead to failures when mitigating errors.
Although we can lower bound the fidelity of the error mitigated state $\rho_\text{EM}$ and $\rho_\text{out}^\text{ideal}$ from Proposition~\ref{prop:fid},
mitigating errors to improve the results still imply a competition between the experiment accuracy and the noise characterization (Proposition~\ref{prop:exp_val} and~\cref{fig:maps}). 
In order to improve experimental readout from EM protocals, the increasing accuracy of the experiments demands better knowledge of device noise, which will translate into expensive procedures and sampling costs on noise characterization.  
The additional overhead due to the need for accurately characterizing noise channels is typically ignored when estimating the cost of QEM~\cite{Takagi2021}, but is nonetheless present for an experimentalist.
These should also be taken into account for future analyses on QEM protocols.

\section{Conclusion and Outlook}
We have examined several aspects of quantum error mitigation in this work.
Using intuitions from classical and quantum communication, we clarify that QEM is fundamentally different from QEC because each method has distinct goals, namely, QEC preserves the physical objects themselves while QEM restores the mathematical descriptions (i.e., density matrices) of the physical objects.
This motivates that inverses of the noise channels play an important role when evaluating the ultimate performance of QEM protocols.
For the case of invertible but CPTP noise where the inverse cannot be physically implemented, we suggest that one may decompose the (HPTP) inverse as multiple implementable CPTP channels, and recombine the results to obtain the inverse using post-processing.

For the case of non-invertible channels, the problem of optimal recovery from the noise is rarely studied in previous literature. 
We explicitly provide a generalized-inverse construction called Drazin-inverse, and prove that the Drazin-inverse of any channel is trace preserving, while another commonly used generalized-inverse called the Moore-Penrose inverse may not be.
Our results provide insights from the point of view of superoperators, opening up new possibilities for applying generalized matrix inverse theory to tackle the above problem. 
The properties of different generalized inverses and when to use them require more investigations.

When the noise channels are invertible, the improvement from EM protocols is constrained by our knowledge about the noise in the device of interest. 
The gap between $\{\tN_i\}$ and $\{\mc{N}_i\}$ can bound the fidelity between the ideal state and the error mitigated state.
A sufficient condition is derived for guaranteeing an improvement, which is motivated from the idea that the accuracy of noise characterization needs to be higher than that of the experiment.
Therefore, more accurate experiments require higher overheads during noise characterization.
We derive a first order approximation to the overall error in multi-layer circuits, which can be used to estimate the minimal cost associated with learning about the noise before applying EM procedures.
We also show that if the mismatch between assumed and real noise channels is too large, error mitigation may yield unphysical outcomes and fail.
Overall, our analyses imply that the complications and subtleties when implementing QEM demand more in-depth future studies on several topics, such as optimal decomposition of inverse maps, and more realistic estimates on the cost of accurate QEM.

\begin{acknowledgments}
N.C. thanks Maxwell Fitzsimmons for helpful discussions.
This research was undertaken thanks in part to funding from the Government of Canada through the Natural Sciences and Engineering Research Council of Canada (NSERC).
\end{acknowledgments}

\bibliography{ErrorMitigationInverse.bib}

\newpage
\onecolumngrid
\appendix

\section{The Drazin-inverse of a TP map is also TP}\label{app:tp}

Here we provide the proof for~\cref{thm:tp}

\begin{proof}

Let the Jordan block of eigenvalue $\lambda$ in $J$ be $J_\lambda$, where $J$ is defined in~\cref{eq:decomp}.

The inverse of a $k$ by $k$ Jordan block $J_\lambda$ ($\lambda\neq 0$) of $\Sp{\mc{N}}$ is 
\begin{align*}
J_\lambda^{-1} & = (\lambda I + N)^{-1} = \lambda^{-1}(I - \lambda^{-1}N +\cdots \lambda^{-(k-1)}N^{(k-1)}) \\
& = \lambda^{-1}\left(\sum_{i=0}^{k-1}(-\lambda^{-1}N)^i\right)=: J'_\lambda 
\end{align*}
where $N$ is the $k$ by $k$ nilpotent matrix, $\lambda$ is the eigenvalue.

By the construction (\cref{eq:inverse}), we have
\begin{equation}\label{eq:jordan}
\Sp{\mc{N}^+}Q = QJ',    
\end{equation}
where $Q$ contains the eigenvectors and generalized eigenvectors of $\Sp{\mc{N}}$.
For the particular block that we concern, the corresponding eigenvector and generalized eigenvector is 
\[Q = 
\begin{pmatrix}
\cdots & v^{g_0} & v^{g_1} &\cdots & v^{g_{k-1}} & \cdots
\end{pmatrix}
\]

Let $e_i, i\in\{0,\cdots,k-1\}$ be the standard basis vectors for this block. 
Acting $e_j$ on both sides of~\cref{eq:jordan}, the left hand side is
\[
\Sp{\mc{N}^+} Q e_j = \Sp{\mc{N}^+} v^{g_j},
\]
and the right hand side is
\[
QJ'e_j = Q\lambda^{-1}\left(\sum_{i=0}^{k-1}(-\lambda)^{-i} N^i e_j\right).
\]
It is easy to show that $N^i e_j = e_{j-i}$ for $j\ge i$, and $N^i e_j = 0 \cdot e_{j}$ for $j< i$. 
Therefore
\[
QJ'e_j = Q\lambda^{-1}\left(\sum_{i = 0}^{j}(-\lambda)^ie_{j-i}\right)
= \lambda^{-1}\left(\sum_{i = 0}^{j}(-\lambda)^iv^{g_{j-i}}\right).
\]
Thus
\begin{equation}\label{eq:ek}
\Sp{\mc{N}^+} v^{g_j} = \lambda^{-1}\left(\sum_{i = 0}^{j}(-\lambda)^iv^{g_{j-i}}\right).
\end{equation}

For $\lambda$ not equals to $1$ and $0$, taking trace on both sides of~\cref{eq:ek}, 
\begin{align*}
\str{\Sp{\mc{N}^+} v^{g_{j}}} = \str{\lambda^{-1}\left(\sum_{i = 0}^{j}(-\lambda)^iv^{g_{j-i}}\right)}.
\end{align*}
From~\cref{lemma1}, we know that $\str{v^{g_{j}}} =0$, so the right hand side is also $0$. That is, $\str{\Sp{\mc{N}^+} v^{g_{j}}} = \str{v^{g_j}}$ holds for every $j\in \{0,\cdots,k-1\}$.

When $\lambda = 1$, according to~\cref{lemma2}, we have the same results except for $j = k-1$. 
Now we check $j = k-1$ case for $\lambda = 1$,
\[
\Sp{\mc{N}^+} v^{g_{k-1}} = \left(\sum_{i=0}^{k-1}(-1)^iv^{g_{k-1-i}}\right).
\]
Taking trace on both sides,
\[
\str{\Sp{\mc{N}^+} v^{g_{k-1}}} = \str{\left(\sum_{i=0}^{k-1}(-1)^iv^{g_{k-1-i}}\right)}
\]
From~\cref{lemma2}, we know that all the eigenvector and generalized eigenvectors have trace zero except for the $(k-1)$th one. We have 
\[
\str{\Sp{\mc{N}^+} v^{g_{k-1}}} = \str{v^{g_{k-1}}}
\]

Finally, when $\lambda = 0$, $J'_\lambda = 0_k$, where $0_k$ is the $k$ by $k$ zero matrix. 
%Thus, $\Sp{\mc{N}^+}Qe_j = Qe_{j-1}$, and $\Sp{\mc{N}^+}v^{g_j} = v^{g_{j-1}}$. 
Thus, $\Sp{\mc{N}^+}v^{g_j} = 0\cdot e_j$.
From~\cref{lemma1}, the trace of both sides are zero. 

Now we have proved that the trace of all columns (eigenvectors and generalized eigenvectors) in $Q$ are unchanged under the action of $\Sp{\mc{N}^+}$. 

Since $Q$ is invertible, any $\Sp{\rho}$ can be expanded by columns $v^{\lambda_{j}}_i$ in $Q$. And we have  
\[\str{\mathbf{v}(\mc{N}^+)(\sum_{ij}a_{ij}v^{\lambda_{j}}_i)} = \str{\sum_{ij}a_{ij}\mathbf{v}(\mc{N}^+)(v^{\lambda_{j}}_i)} =  \str{\sum_{ij}a_{ij}v^{\lambda_{j}}_i},\]
Hence, $\mc{N}^+$ is trace preserving.

\end{proof}

\section{The Effect of Imperfect Knowledge about Noise Channels on Fidelity}\label{app:imperfect}

From the main text, we know that 
\[
\rho_\text{EM} = \mc{U}_{n\cdots 1}\circ \tilde{\mc{R}} (\rho_\text{out}^\text{exp}),
\]
\[
\rho_\text{out}^\text{ideal} = \mc{U}_{n\cdots 1}\circ \mc{R} (\rho_\text{out}^\text{exp}),
\]
where $\mc{U}_{n\cdots1} \coloneqq \mc{U}_n\circ \cdots \circ \mc{U}_1$ is the ideal set of circuits (See~\cref{fig:maps}).

Imperfect knowledge about $\mc{N}_i$ leads to imperfect inverse $\tilde{N}_i^{-1}$. Let $\tilde{N}_i^{-1} = \mc{N}_i^{-1} + \Delta \mc{N}_i^{-1}$, where $\mc{N}_i^{-1}$ is the perfect inverse of $\mc{N}_i$.

Let $\Delta \rho_\text{EM} \coloneqq \rho_\text{EM} - \rho_\text{out}^\text{ideal} $, then
\begin{equation}
\Delta \rho_\text{EM} = \mc{U}_{n\cdots1}\circ\left[\tilde{\mc{R}} - \mc{R}\right](\rho_\text{out}^\text{exp})
= \mc{U}_{n\cdots1}\circ \Delta \mc{N} (\rho_\text{out}^\text{exp})
.
\end{equation}

If we only consider the first order approximation $\Delta \mc{N}^{(1)}$ of $\Delta \mc{N}$, the first order correction term $\Delta \rho_\text{EM}^{(1)}$ would be
\begin{align}\label{eq:delta_rho}
%\begin{equation}
\Delta \rho_\text{EM}^{(1)} & = \mc{U}_{n\cdots1}\circ \Delta \mc{N}^{(1)} \circ \mc{U}_\text{exp} (\rho_\text{in})\\
%&= \mc{U}_{n\cdots2}\circ\left(\sum_{i=1}^{n}\mc{N}_1^{-1}\circ \mc{U}_2^\dag\circ \mc{N}_2^{-1}\circ \cdots\circ \mc{U}^\dag_{i}\circ\Delta \mc{N}_i^{-1}\circ \mc{U}^\dag_{i+1}\circ \cdots \mc{U}^\dag_{n}\circ \mc{N}_{n}^{-1}\right)\circ \mc{U}_\text{exp}(\rho_\text{in})\\
& = \mc{U}_{n\cdots 1}\circ
%\left(\sum_{i=1}^{n}\mc{U}_1^\dag\circ \tN_1^{-1}\circ \mc{U}_2^\dag\circ \tN_2^{-1}\circ \cdots\circ \mc{U}^\dag_{i}\circ\Delta \tN_i^{-1}\circ \mc{U}^\dag_{i+1}\circ \cdots \mc{U}^\dag_{n}\circ \tN_{n}^{-1}\right)\circ \mc{U}_\text{exp}(\rho_\text{in})
\left(\sum_{i=1}^{n}\mc{U}_1^\dag\circ \tN_1^{-1}\cdots\circ \mc{U}^\dag_{i}\circ\Delta \tN_i^{-1}\circ \cdots \mc{U}^\dag_{n}\circ \tN_{n}^{-1}\right)\circ \mc{U}_\text{exp}(\rho_\text{in}),
\end{align}
%\end{equation}
where $\mc{U}_\text{exp}\coloneqq \mc{N}_n\circ \mc{U}_n\circ\cdots \circ \mc{N}_1\circ \mc{U}_1$ is the actual experimental operator. 

Then $\rho_\text{EM} \approx \rho_\text{out}^\text{ideal} + \Delta \rho_\text{EM}^{(1)}$. 
The first order approximation of $F(\rho_\text{EM},\rho_\text{out}^\text{ideal})$, the fidelity between $\rho_\text{EM}$ and $\rho_\text{out}^\text{ideal}$, is
\[
F(\rho_\text{EM},\rho_\text{out}^\text{ideal}) \approx F^{(1)}(\rho_\text{EM},\rho_\text{out}^\text{ideal}) \coloneqq F(\rho_\text{out}^\text{ideal}+ \Delta \rho_\text{EM}^{(1)},\rho_\text{out}^\text{ideal}) = \|\sqrt{\rho_\text{out}^\text{ideal}+ \Delta \rho_\text{EM}^{(1)} }\sqrt{\rho_\text{out}^\text{ideal}}\|_\text{tr}.
\]
By the Fuchs--van de Graaf inequalities, 
\begin{equation}\label{eq:FDG}
[1-D(\Delta \rho_\text{EM}^{(1)})]^2\le F^{(1)}(\rho_\text{EM},\rho_\text{out}^\text{ideal}) \le 1- D^2(\Delta \rho_\text{EM}^{(1)}),
\end{equation}
where $D(\cdot) \coloneqq \frac{1}{2}\|\cdot\|_\text{tr}$ is the trace distance, and $\|\cdot\|_\text{tr}$ is the trace norm. It is also known that $\|A\|_{F} \leq\|A\|_\text{tr} \leq \sqrt{r}\|A\|_{F}$, where $\|A\|_{F}$ is the Frobenius norm which equals to $\|\Sp{A}\|$. The norm $\|\cdot\|$ is the 2-norm.
Since
\[
\|\Delta \rho_\text{EM}^{(1)}\|_{F} = \|\Sp{\Delta \rho_\text{EM}^{(1)}}\| = \|\mathbf{v}(\mc{U}_{n\cdots1})\mathbf{v}(\Delta \mc{N}^{(1)})\mathbf{v}(\mc{U}_\text{exp})\Sp{\rho_\text{in}}\| = \|\mathbf{v}(\mc{U}_{n\cdots1})\mathbf{v}(\Delta \mc{N}^{(1)})\Sp{\rho_\text{out}^\text{exp}}\|,
\]
therefore we can bound $\|\Delta \rho_\text{EM}^{(1)}\|_{F}$ by
\begin{equation}\label{eq:boundF}
l_U\cdot \|\mathbf{v}(\Delta \mc{N}^{(1)})\Sp{\rho_\text{out}^\text{exp}}\| \le \|\Delta \rho_\text{EM}^{(1)}\|_{F}\le \|\mathbf{v}(\mc{U}_{n\cdots1})\|
\cdot\|\mathbf{v}(\Delta \mc{N}^{(1)})\|\cdot\|\Sp{\rho_\text{out}^\text{exp}}\|
\end{equation}
where $l_U \coloneqq \inf_{\|x\| = 1} \|\mathbf{v}(\mc{U}_{n\cdots1}) x\|$ is the lower Lipschitz constant of the superoperator of the ideal circuits. Notice that, on the right hand side of~\cref{eq:boundF}, $\|\mathbf{v}(\mc{U}_{n\cdots1})\|$ and $\|\Sp{\rho_\text{out}^\text{exp}}\|$ are known for a given experiment. Denote $(\|\mathbf{v}(\mc{U}_{n\cdots1})\|\cdot\|\Sp{\rho_\text{out}^\text{exp}}\|)$ as $C_\text{exp}$.
From~\cref{eq:FDG}, we know the fidelity between the mitigated state and the ideal state is bounded by

\[
\left(1 - \frac{1}{2}\sqrt{d} C_\text{exp}\|\mathbf{v}(\Delta \mc{N}^{(1)})\|\right)^2\le F^{(1)}(\rho_\text{EM},\rho_\text{out}^\text{ideal}) \le 1-\frac14 \left(l_U\cdot \|\mathbf{v}(\Delta \mc{N}^{(1)})\Sp{\rho_\text{out}^\text{exp}}\|\right)^2.
\]

%And note that 
In addition, the norm of $\Sp{\Delta \mc{N}^{(1)}}$ satisfies that 
\begin{equation}\label{eq:expending}
\left\|\Sp{\Delta \mc{N}^{(1)}}\right\| \le 
\prod_{k=1}^n\left\|\Sp{\mc{U}_k^\dag}\right\|\cdot
\sum_{i=1}^{n} \left[\left\|\Sp{\Delta \mc{N}_i^{-1}}\right\| \prod_{\substack{j\in\{1,\cdots,n\}\\j\neq i}}\left\|\Sp{\tN_j^{-1}}\right\|\right],
\end{equation}
where $\mc{U}^\dag_k$ and $\tN_j^{-1}$ are known for a given EM tasks. The error on each inverse $\Delta \mc{N}_i^{-1}$ can be exposed at the lower bound of the fidelity. And by counting the sampling cost of getting $\Delta\mc{N}^{-1}$, one can bound the fidelity from the sampling cost.

\section{A Sufficient Condition on Improving Expectation Values}\label{app:exp_val}

The goal of error mitigation on the expectation value of an observables $A$ is 
\begin{equation}\label{eq:em_val}
|\Tr(\rho_\text{EM} A) - \Tr(\rho_\text{out}^\text{ideal}A)|
\le|\Tr(\rho_\text{out}^\text{ideal} A) - \Tr(\rho_\text{out}^\text{exp} A)|.
\end{equation}

The left hand side of~\cref{eq:em_val} is 
\begin{equation}\label{eq:em_val_diff}
\left|\Tr[(\rho_\text{EM}  - \rho_\text{out}^\text{ideal})A]\right|
= \left|\Tr(\Delta \rho_\text{EM} A)\right| 
= \left|\Tr(\mc{U}_{n\cdots1}\circ \Delta \mc{N} (\rho_\text{out}^\text{exp}) \cdot A)\right|
%= \left|\Tr(\mc{U}_{n\cdots1}\circ \Delta \mc{N}^{(1)} (\rho_\text{out}^\text{exp}) \cdot A)\right|.
= \left|\left\langle \Sp{\mc{U}^\dag_{n\cdots1}}\Sp{A^\dag},\Sp{\Delta \mc{N}}\Sp{\rho_\text{out}^\text{exp}}\right\rangle\right|.
\end{equation}

The right hand side of~\cref{eq:em_val} equals to 
\begin{align}\label{eq:exp_val_diff}
\left|\Tr[(\rho_\text{out}^\text{ideal} -\rho_\text{out}^\text{exp}) A]\right| 
& = \left|\Tr[(\mc{U}_{n\cdots 1}\circ \mc{R} - \mc{I}) (\rho_\text{out}^\text{exp}) \cdot A]\right|
=\left|\Tr[\mc{U}_{n\cdots 1}\circ [\mc{R} - \mc{U}^\dag_{1\cdots n}] (\rho_\text{out}^\text{exp})
\cdot A]\right| \nonumber \\
& = \left|\left\langle \Sp{\mc{U}^\dag_{n\cdots1}}\Sp{A^\dag},\Sp{\mc{R} - \mc{U}^\dag_{1\cdots n}}\Sp{\rho_\text{out}^\text{exp}}\right\rangle\right|
\end{align}

It is difficult to draw conclusions directly from~\cref{eq:em_val_diff} and~\cref{eq:exp_val_diff} since $\Sp{\Delta \mc{N}}$ and $\Sp{\mc{U}^\dag\mc{N}_{1\cdots n}^{-1} - \mc{U}^\dag_{1\cdots n}}$ can be arbitrary.
However, 
\[
\left|\left\langle \Sp{\mc{U}^\dag_{n\cdots1}}\Sp{A^\dag},\Sp{\Delta \mc{N}}\Sp{\rho_\text{out}^\text{exp}}\right\rangle\right| 
\le \left\|\Sp{\mc{U}^\dag_{n\cdots1}}\Sp{A^\dag}\right\| \left\|\Sp{\rho_\text{out}^\text{exp}}\right\|\left\|\Sp{\Delta \mc{N}}\right\|,
\]

\[
\left|\left\langle \Sp{\mc{U}^\dag_{n\cdots1}}\Sp{A^\dag},\Sp{\mc{R} - \mc{U}^\dag_{1\cdots n}}\Sp{\rho_\text{out}^\text{exp}}\right\rangle\right|
\ge \left\|\Sp{\mc{U}^\dag_{n\cdots1}}\Sp{A^\dag}\right\| \left\|\Sp{\rho_\text{out}^\text{exp}}\right\|\inf_{\|x\| = 1}\left\| \Sp{\mc{R} - \mc{U}^\dag_{1\cdots n}}x \right\|.
\]
Therefore, if 
\[
\left\|\Sp{\Delta \mc{N}}\right\| \le \inf_{\|x\| = 1}\left\| \Sp{\mc{R} - \mc{U}^\dag_{1\cdots n}}x \right\|,
\]
then \cref{eq:em_val} is guaranteed. 
This means that the EM process will improve the expectation value for any observable $A$ and any desired circuit $\mc{U}_{n\cdots 1}$ when the above is satisfied.
This is a stringent requirement, since if $\Sp{\mc{R} - \mc{U}^\dag_{1\cdots n}}$ has a nontrivial null space, then the right hand side equals to 0.
This implies $\tN_i = \mc{N}_i$ for $\forall i\in \{1,\cdots,n\}$, i.e., the estimation on all noise channels $\tN_i$ must be perfect.

\section{Examples of Noise Channel Mismatching}\label{app:depolarizing}

The Kraus representation of $\mc{N}$ and $\mc{D}$ are
\[
\mc{N} : \left\{\sqrt{p_1}I, \sqrt{p_2} X, \sqrt{p_3} Y, \sqrt{(1-p_1-p_2-p_3)}Z\right\};
\]
\[\mc{D} : \left\{\sqrt{1-\frac{3\lambda}{4}}I, \sqrt{\frac{\lambda}{4}} X, \sqrt{\frac{\lambda}{4}} Y, \sqrt{\frac{\lambda}{4}}Z \right\}.
\]

For a given set of $\{p_1,p_2,p_3\}$, what is the optimal $\lambda$ to minimize $\|\mc{N}-\mc{D}\|$ of a chosen norm $\|\cdot\|$?
One approach is to write down a matrix representation of $\mc{N}$ and $\mc{D}$, then solve $\lambda$ by minimizing $\|N-D\|$ for a particular choice of norm.
For different representations and/or norms, the optimization outcome could be different.
The optimal $\lambda$ will bound the distance $\|\mc{N}-\mc{D}\|$ from below for any possible experimental implementation for this particular norm $\|\cdot\|$.

As mentioned in the main text, the two vectors, $\vec n \coloneqq (\sqrt{p_1}, \sqrt{p_2}, \sqrt{p_3}, \sqrt{(1-p_1-p_2-p_3)})$ and $\vec d \coloneqq\left(\sqrt{1-\frac{3\lambda}{4}}, \sqrt{\frac{\lambda}{4}}, \sqrt{\frac{\lambda}{4}}, \sqrt{\frac{\lambda}{4}}\right)$, are also representations for $\mc{N}$ and $\mc{D}$ respectively. Since $\vec n$ and $\vec d$ are normalized, minimizing the distance between $\mc{N}$ and $\mc{D}$ is equivalent to maximizing $\vec n \cdot \vec d$. i.e. 
\[\max_{\lambda\in [0,1]} \left\{\sqrt{p_1 (1-\frac{3\lambda}{4})} + [\sqrt{p_2}+\sqrt{p_3}+\sqrt{(1-p_1-p_2-p_3)}]\sqrt{\frac{\lambda}{4}}\right\}.
\]

This can be solved by taking the derivative of the expression, and setting it to be zero.
The result is 
\begin{equation}\label{eq:lambda_max}
\lambda_{max} = \frac{[\sqrt{p_2}+\sqrt{p_3}+\sqrt{(1-p_1-p_2-p_3)}]^2 p_1}{\frac94 p_1^2 + \frac34 p_1 [\sqrt{p_2}+\sqrt{p_3}+\sqrt{(1-p_1-p_2-p_3)}]^2}, \text{ or }\lambda = 1,  \text{ or }\lambda = 0.
\end{equation}

\vspace{1cm} 

The superoperators of $\mathcal{N}$ and $\mc{D}$ are
\[\Sp{\mc{N}} = p_1
\begin{pmatrix}
1 & 0 & 0 & 0\\
0 & 1 & 0 & 0\\
0 & 0 & 1 & 0\\
0 & 0 & 0 & 1\\
\end{pmatrix}
+ p_2 
\begin{pmatrix}
0 & 0 & 0 & 1\\
0 & 0 & 1 & 0\\
0 & 1 & 0 & 0\\
1 & 0 & 0 & 0\\
\end{pmatrix}
+p_3 
\begin{pmatrix}
0 & 0 & 0 & 1\\
0 & 0 & -1 & 0\\
0 & -1 & 0 & 0\\
1 & 0 & 0 & 0\\
\end{pmatrix}
+(1-p_1-p_2-p_3)
\begin{pmatrix}
1 & 0 & 0 & 0\\
0 & -1 & 0 & 0\\
0 & 0 & -1 & 0\\
0 & 0 & 0 & 1\\
\end{pmatrix},
\]
\[\Sp{\mc{D}} = (1-\frac{3\lambda}{4})
\begin{pmatrix}
1 & 0 & 0 & 0\\
0 & 1 & 0 & 0\\
0 & 0 & 1 & 0\\
0 & 0 & 0 & 1\\
\end{pmatrix}
+\frac{\lambda}{4} 
\begin{pmatrix}
0 & 0 & 0 & 1\\
0 & 0 & 1 & 0\\
0 & 1 & 0 & 0\\
1 & 0 & 0 & 0\\
\end{pmatrix}
+\frac{\lambda}{4}
\begin{pmatrix}
0 & 0 & 0 & 1\\
0 & 0 & -1 & 0\\
0 & -1 & 0 & 0\\
1 & 0 & 0 & 0\\
\end{pmatrix}
+\frac{\lambda}{4}
\begin{pmatrix}
1 & 0 & 0 & 0\\
0 & -1 & 0 & 0\\
0 & 0 & -1 & 0\\
0 & 0 & 0 & 1\\
\end{pmatrix}.
\]

Even with the optimal $\lambda$ in~\cref{eq:lambda_max}, when $p_2, p_3$ and $1-p_2-p_3$ are not equal to each other, the distance between $\mc{N}$ and $\mc{D}$ is not zero.

The following are two examples of different sets of $\{p_i\}$. 
\begin{enumerate}%[label=(\alph*)]
	\item When $p_1 = p_3 =  0$ and $p_2 = 1$, the optimal $\lambda_\text{max}$ is $1$. Therefore
	\[
	\Sp{\mc{N}} =
	\begin{pmatrix}
	0 & 0 & 0 & 1\\
	0 & 0 & 1 & 0\\
	0 & 1 & 0 & 0\\
	1 & 0 & 0 & 0\\
	\end{pmatrix}, \text{ }
	\Sp{\mc{D}}  =\frac12
	\begin{pmatrix}
	1 & 0 & 0 & 1\\
	0 & 0 & 0 & 0\\
	0 & 0 & 0 & 0\\
	1 & 0 & 0 & 1\\
	\end{pmatrix}
	\]
	
	In this case, the estimated $\mc{D}$ is non-invertible while $\mc{N}$ is invertible. 
	Any generalized inverse $\mc{D}$ will definitely worsen the outcomes. 
	
	\item When $p_1 = \frac12$ and $p_2 = p_3 = 0$, we have $\lambda_\text{max} = \frac13$ according to~\cref{eq:lambda_max}. 
	\begin{figure}[ht!]
	\centering
		\begin{subfigure}{.49\linewidth}
		\includegraphics[width = .9\linewidth]{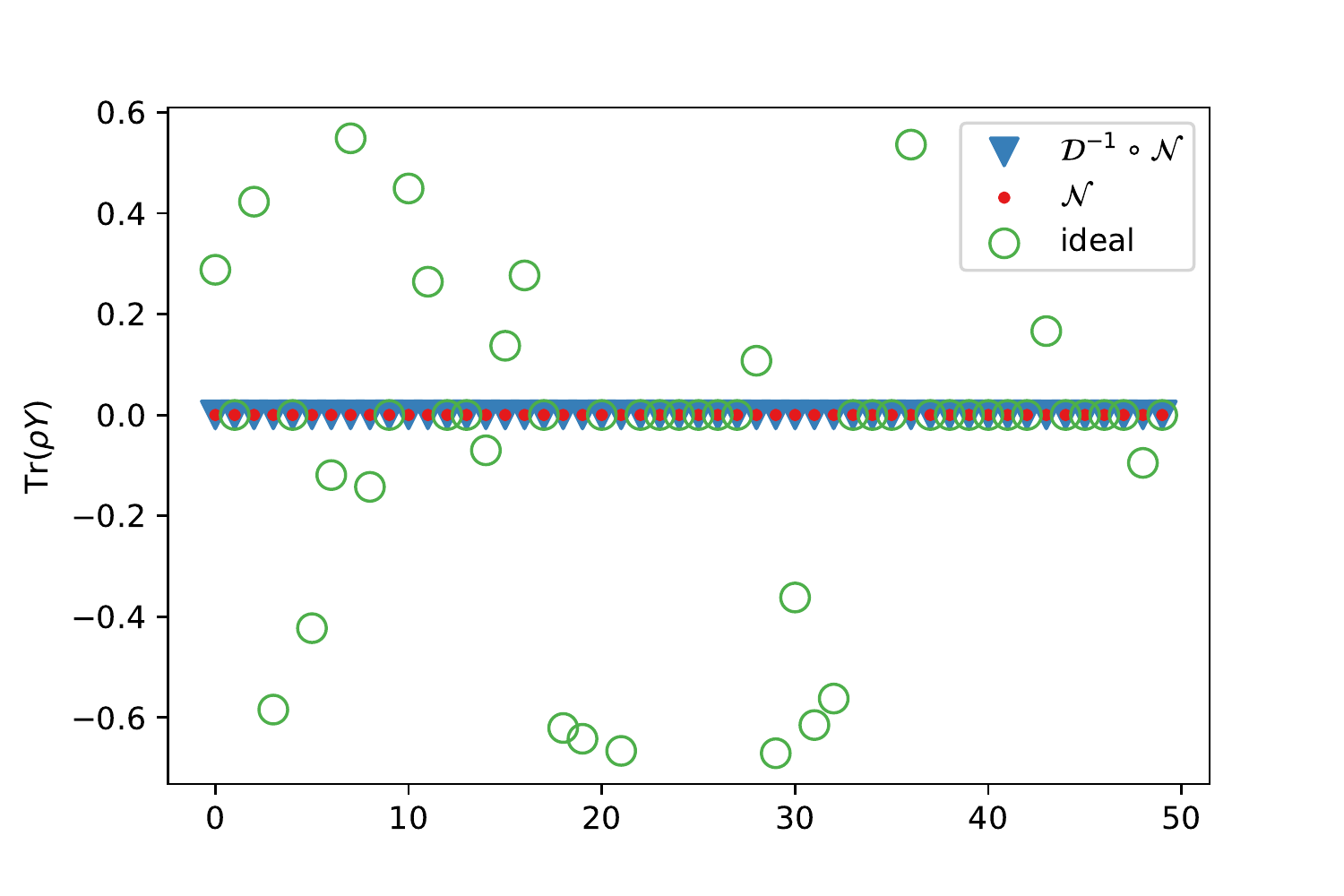}
		\caption{Expectation values of Pauli $Y$ for 50 randomly generated states.}\label{fig:appendixE_Y}
		\end{subfigure}
		\begin{subfigure}{.49\linewidth}
		\includegraphics[width = .9\linewidth]{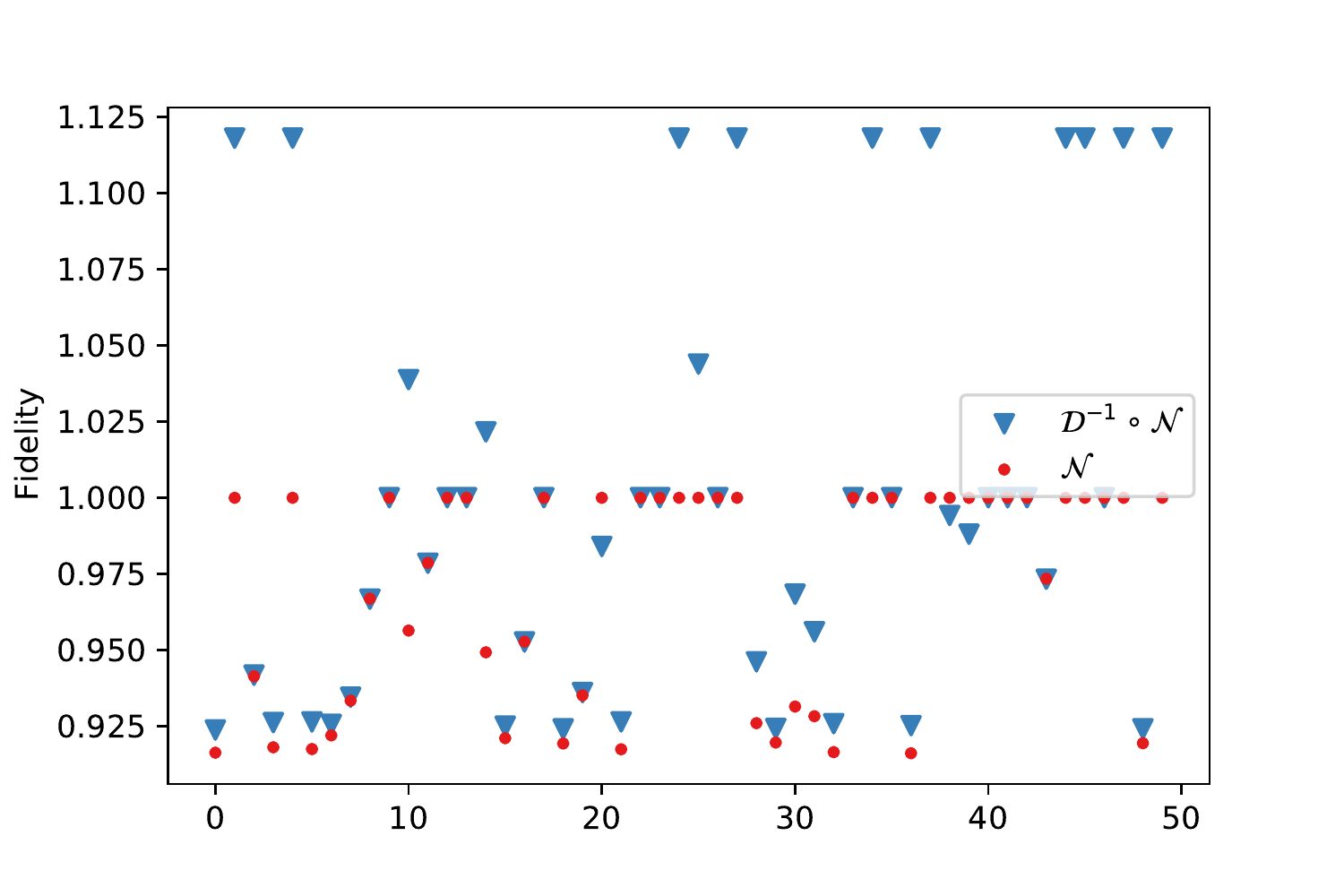}
		\caption{Fidelities of 50 randomly generated states.}\label{fig:app_e_fid}
		\end{subfigure}
	\caption{
		Effects of applying the mismatched noise channel, $\mc{D}^{-1}$, to 50 randomly generated noisy output states.
		The $x$-axis is a dummy label for the tested states. 
		Because the channel $\mc{D}^{-1}\circ\mc{N}$ is not physical (not CP), the ``mitigated'' outputs $\mc{D}^{-1}\circ\mc{N}(\rho)$ are not valid quantum states. In this case the fidelity is no longer a good metric for distinguishing two ``states''.}\label{fig:app_e} 
	\end{figure}

	\[
	\Sp{\mc{N}} =
	\begin{pmatrix}
	1 & 0 & 0 & 0\\
	0 & 0 & 0 & 0\\
	0 & 0 & 0 & 0\\
	0 & 0 & 0 & 1\\
	\end{pmatrix}, \text{ }
	\Sp{\mc{D}} = \frac16
	\begin{pmatrix}
	5 & 0 & 0 & 1\\
	0 & 4 & 0 & 0\\
	0 & 0 & 4 & 0\\
	1 & 0 & 0 & 5\\
	\end{pmatrix}
	\]
	
	The inverse of $\mc{D}$ is 
	\[\Sp{\mc{D}^{-1}} =\frac14
	\begin{pmatrix}
	5 & 0 & 0 & -1\\
	0 & 6 & 0 & 0\\
	0 & 0 & 6 & 0\\
	-1 & 0 & 0 & 5\\
	\end{pmatrix}.
	\]
	Therefore,
	\[
	\Sp{\mc{N}^\text{re}} \coloneqq \Sp{\mc{D}^{-1}}\Sp{\mc{N}} = \frac14
	\begin{pmatrix}
	5 & 0 & 0 & -1\\
	0 & 0 & 0 & 0\\
	0 & 0 & 0 & 0\\
	-1 & 0 & 0 & 5\\
	\end{pmatrix}
	\]
	This resulting channel $\mc{N}^\text{re}$ has eigenvalues $\{\frac32,1,0,0\}$, which will worsen the outcome.
	In~\cref{fig:app_e}, we tested 50 randomly generated quantum state $\rho$ for this example.
	\cref{fig:appendixE_Y} shows the expectation value of Pauli $Y$ for these 50 states.
	Since the expectation value $\Tr(Y\rho)$ is erased by the noise channel $\mc{N}$, $\mc{D}^{-1}$ has no effect on improving $\Tr(Y\rho)$.
	In comparison, for $\Tr(Z\rho)$ in~\cref{fig:app_e_val}, the channel $\mc{D}^{-1}$ has made the outcome worse.
	\cref{fig:app_e_fid} shows the fidelities $F(\mc{N}(\rho),\rho)$ and $F(\mc{N}^\text{re}(\rho),\rho)$. Since $\mc{D}^{-1}$ is non-CP, the outputs $\mc{D}^{-1}\circ\mc{N}(\rho)$ are not valid quantum states anymore. The fidelity function does not always smaller than 1, thus is no longer a good metric.
	This explains why the recovery $\mc{D}^{-1}$ does not improve any expectation value but seems to have higher fidelities.

\end{enumerate}

\end{document}